\documentclass[acmlarge]{acmart}

\usepackage{booktabs} % For formal tables
\usepackage{hyperref}   %gws 20170930
\usepackage{url}  %gws 20170930
\definecolor{myred}{RGB}{255,204,204}
\definecolor{myblue}{RGB}{0,200,255}
\definecolor{mygreen}{RGB}{80,220,80}
\definecolor{thedarkblue}{RGB}{0,0,120} %104} % 180
\definecolor{mydarkblue}{rgb}{0,0.08,0.45} %ICML dark blue
\hypersetup{%
	colorlinks=true,
	linkcolor=mygreen,
	citecolor=mydarkblue,
	filecolor=mydarkblue,
	urlcolor=myred}
\usepackage{multirow}  %gws 20170930
\usepackage{color}  %gws 20170930
\newtheorem{strategy}{Strategy}    % gws
\usepackage[ruled]{algorithm2e} % For algorithms

\SetAlFnt{\small}
\SetAlCapFnt{\small}
\SetAlCapNameFnt{\small}
\SetAlCapHSkip{0pt}
\IncMargin{-\parindent}

\acmJournal{TSC}
\acmVolume{9}
\acmNumber{4}
\acmArticle{39}
\acmYear{2018}
\acmMonth{3}
\acmArticleSeq{11}

\setcopyright{usgovmixed}
\acmDOI{0000001.0000001}

\received{February 2018}
\received{March 2018}
\received[accepted]{XXX 2019}

\begin{document}
% Title portion
\title{Beyond Frequency: Utility Mining with Varied Item-Specific Minimum Utility}
%\titlenote{We can add a note to the title}

\author{Wensheng Gan}
\affiliation{%
	\institution{Harbin Institute of Technology (Shenzhen)}
	\city{Shenzhen}
	\country{China}
}
\email{wsgan001@gmail.com}

\author{Jerry~Chun-Wei~Lin}
%\authornote{This is the corresponding author}
\affiliation{%
	\institution{Western Norway University of Applied Sciences}
	\city{Bergen}
	\country{Norway}
}
\email{jerrylin@ieee.org}

\author{Philippe~Fournier-Viger}
\affiliation{%
	\institution{Harbin Institute of Technology (Shenzhen)}
	\city{Shenzhen}
	\country{China}	
}
\email{philfv@hitsz.edu.cn}

\author{Han-Chieh Chao}
\affiliation{%
  \institution{National Dong Hwa University}
  \city{Hualien}
  \country{Taiwan}
}
\email{hcc@ndhu.edu.tw}

\author{Philip S. Yu}
\affiliation{%
	\institution{University of Illinons at Chicago}
	\city{Chicago}
	\country{USA}
}
\email{psyu@uic.edu}

\begin{abstract}

Consumer behavior plays a very important role in economics and targeted marketing. However, understanding economic consumer behavior is quite challenging, such as finding credible and reliable information on product profitability. Different from frequent pattern mining, utility-oriented mining integrates utility theory and data mining. Utility mining is a useful tool for understanding economic consumer behavior. Traditional algorithms for mining high-utility patterns (HUPs) applies a single/uniform minimum high-utility threshold (\textit{minutil}) to obtain the set of HUPs, but in some real-life circumstances, some specific products may bring lower utilities compared with others, but their profit may offer some vital information. However, if \textit{minutil} is set high, the patterns with low \textit{minutil} are missed; if \textit{minutil} is set low, the number of patterns becomes unmanageable. In this paper, an efficient one-phase utility-oriented pattern mining algorithm, called HIMU, is proposed for mining HUPs with varied item-specific minimum utility. A novel tree structure called a multiple item utility set-enumeration tree (MIU-tree), the global sorted  and the conditional downward closure properties are introduced in HIMU. In addition, we extended the compact utility-list structure to keep the necessary information, and thus this one-phase HIMU model greatly reduces the computational costs and memory requirements. Moreover, two pruning strategies are then extended to enhance the performance. We conducted extensive experiments in several synthetic and real-world datasets; the results indicates that the designed one-phase HIMU algorithm can address the ``\textit{rare item problem}" and has better performance than the state-of-the-art algorithms in terms of runtime, memory usage, and scalability. Furthermore, the enhanced algorithms outperform the non-optimized HIMU approach.

\end{abstract}


%\ccsdesc[500]{Information systems~Association rules}

%H.2.8 [Database Applications]: Data mining



\maketitle

\renewcommand{\shortauthors}{W. Gan et al.}

\section{Introduction}
\label{sec:introduction}

Consumer behavior is an interesting topic in the fields of economy and targeted marketing. However, understanding economic consumer behavior is quite challenging, such as finding credible and reliable information on products' profitability. Data mining techniques can provide a powerful capability for finding the useful and meaningful information from different types of data \cite{agrawal1995mining,agrawal1993mining,han2004mining,srikant1996mining,gan2017data}. For data mining, the association rule is the rule most commonly used to represent the relationships among item/sets from transactional data. Frequent pattern mining (FPM) or association rule mining (ARM) has been extensively discussed  and applied in a wide range of applications \cite{agrawal1994fast,geng2006interestingness,agrawal1993mining,chen1996data,han2004mining}. However, FPM and ARM \cite{agrawal1993mining,chen1996data,han2004mining} concern only the co-occurrence frequency of patterns in transactional databases. Other implicit factors, such as importance, interestingness, weight, risk, and profit, do not take into account of FPM. Additionally, each item is associated with the same significance in FPM; the actual significant item/sets cannot be easily discovered. In the market basket analytics, for example, both of two group/combination products \{\textit{bread}, \textit{milk}\} and \{\textit{bread}, \textit{milk}, \textit{beer}\} are concerned as the frequent products since their frequencies are the same and greater than the minimum support threshold. If the profit of combination products \{\textit{bread}, \textit{milk}\} is \$3 and that of \{\textit{bread}, \textit{milk}, \textit{beer}\} is \$10, there is no doubt that \{\textit{bread}, \textit{milk}, \textit{beer}\} can bring higher profit than that of combination products \{\textit{bread}, \textit{milk}\}.

Considering consumer choice and the preferences constraint, several crucial factors for the task of commerce have been advanced (w.r.t. maximal profit/utility): (i) rational behavior; (ii) preferences are known and measurable; (iii) inventory management; and (iv) prices. To obtain higher profit from the products, the decision-maker should determine the most profitable products from historical records. Consequently, the utility-based pattern mining framework called high-utility pattern/itemset mining (HUIM) was proposed. It integrates the utility theory \cite{li2011towards} from Microeconomics and pattern mining \cite{agrawal1995mining,han2004mining} from data mining. In general, HUIM aims to find high profitable products/itemsets, as it concerns two factors such as unit profit and quantity of the items. Accordingly, the utility (i.e., importance, interestingness, risk, or profit) of each product/itemset can be predefined based on users' priori knowledge. A product is considered to be a high-utility itemset (HUI) if its total utility in a database is greater or equal to the user-specified minimum utility threshold. HUIM has become a useful tool for understanding economic consumer behavior with profitable products, and many approaches have been designed \cite{ahmed2009efficient,chan2003mining,tseng2010up,yao2006mining,liu2005two}. The previous works consider, however, a single minimum utility threshold to discover the complete set of HUIs. In some real-world situations, items may have hierarchy attributes and unit-contributed utilities can be different, and as a result, all of these approaches do not support hierarchies. Using a single and uniform threshold to examine the utility of all item/sets in a database is inappropriate in real-life situations, since each item has different importance and should not be considered as the same. Using the single minimum utility threshold (\textit{minutil}) implies that all patterns bring with them similar utilities beyond \textit{minutil} in the database. However, this may not be workable in practical applications. In real-world situations, each item may have different nature, frequency, importance, and different exhibition utility. It is thus difficult to carry out a fair measurement by only considering a single minimum utility threshold for all items. When \textit{minutil} is set too high, the useful patterns with low \textit{minutil} are missed. When \textit{minutil} is set too low, enormous patterns are then produced.
	
	\begin{table}[!htbp]
		\centering
		\small
		\caption{A transaction database w.r.t. shopping behavior.} 
		\label{table_example1}
		\begin{tabular}{|c|c|c|c|}
			\hline
			\textbf{TID} & \textbf{Time} & \textbf{Customer ID}  & \textbf{Event (product, quantity, total profit)} \\ \hline
			$T_1$   &   03-06-2017 10:40:30   &   $C_1$  &  \textit{milk}, 1, \$2; \textit{bread}, 2, \$6; \textit{beer}, 1, \$6; \textit{cheese}, 1, \$3;  \textit{cereal}, 2, \$3 	    \\ 
			$T_2$   &   03-06-2017 10:59:12   &   $C_2$  &  \textit{foodProcessor}, 2, \$24	    \\ 
			$T_3$   &   03-06-2017 11:08:45   &   $C_3$  &  \textit{milk}, 3, \$6;  \textit{bread}, 2, \$6; \textit{beer}, 1, \$6	    \\ 
			$T_4$   &   03-06-2017 11:40:00   &   $C_1$  &  \textit{milk}, 1, \$2; \textit{cereal}, 1, \$1.5; \textit{cheese}, 1, \$3	    \\ 
			$T_5$   &   03-06-2017 12:55:30   &   $C_3$  &  \textit{foodProcessor}, 1, \$12; \textit{cookingPan}, 1, \$16	    \\ 
			$T_6$   &   03-06-2017 14:38:58   &   $C_5$  &  \textit{bread}, 1, \$3; \textit{milk}, 1, \$2 \\ 
			...   &   ...   &   ...  & ......	    \\ 
			$T_{10}$   &   05-06-2017 15:30:00   &   $C_1$  &  \textit{bread}, 1, \$3; \textit{butter}, 2, \$4; \textit{oatmeal}, 1, \$10	    \\ 						
			\hline
		\end{tabular}
	\end{table}

For example, Table \ref{table_example1} is a simple retail store's database which contains customers' shopping records, including transaction ID (TID), occur time, customer ID, and event (w.r.t. purchase product, quantity and total profit), and so on. In this transaction data, in order to find profitable patterns involving those low-profitable purchased products such as food processor and cooking pan (they generate more profits per item), we need to set the \textit{minutil} to very high (e.g., \$30). We may find the following useful pattern from Table \ref{table_example1}: \{\textit{foodProcessor}, \textit{cookingPan}\} [utility = \$28]. However, this high \textit{minutil} may cause the following high profitable patterns cannot be found: \{\textit{bread}, \textit{milk}, \textit{beer}\} [utility = \$18]. If the \textit{minutil} is set very low, some meaningless low-profitable patterns would be found. This dilemma is called the ``\textit{rare item problem}" \cite{weiss2004mining}. Mining frequent itemsets or association rules based on multiple minimum support thresholds have been extensively discussed, such as MSApriori \cite{liu1999mining}, CFP-Growth \cite{hu2006mining}, CFP-Growth++ \cite{kiran2011novel}, REMMAR \cite{liu2011discovering}, and FQSP-MMS \cite{huang2013discovery}. However, all of these existing methods for support-based pattern mining cannot be applied into the utility-based HUIM problem. Up to now, fewer works have been discussed to handle the problem of HUIM with multiple minimum utility thresholds (MMU model). To the best of our knowledge, HUI-MMU and its improved HUI-MMU$_{TID}$ algorithms \cite{lin2016efficient} were the first algorithms focusing on the topic for mining HUIs with multiple minimum utility thresholds. Those algorithms apply, however, the generate-and-test mechanism to discover the set of HUIs in a level-wise manner. More computations and memory leakage may thus occur while the transaction size is extremely high. 

To enhance the effectiveness and efficiency of HUIM, an efficient utility-based pattern algorithm named \textbf{\underline{\textit{H}}}igh-efficient utility mining with varied \textbf{\underline{\textit{I}}}tem-specific \textbf{\underline{\textit{M}}}inimum \textbf{\underline{\textit{U}}}tility (HIMU) is designed. An MIU-tree is also developed to enhance and improve the mining performance to discover the complete set of HUIs. Based on the developed MIU-tree and a sorted downward closure property of HUIs under multiple minimum utility thresholds, the HIMU algorithm can efficiently obtain the set of HUIs with several pruning strategies. Several improved algorithms are further developed to early prune the huge number of redundant or invalid itemsets. The major contributions of this paper are as follows:

\begin{itemize}
	\item An efficient one-phase algorithm named \textbf{\underline{\textit{H}}}igh-efficient utility mining with varied \textbf{\underline{\textit{I}}}tem-specific \textbf{\underline{\textit{M}}}inimum \textbf{\underline{\textit{U}}}tility (HIMU) is proposed for revealing more useful and meaningful HUIs with varied item-specific minimum utility thresholds. HIMU is more flexible and realistic than the traditional HUIM algorithms that use a single and uniform minimum utility threshold. To the best of our knowledge, this is the first one-phase algorithm which utilizes vertical data structure to address the HIMU problem. 
	
	\item A novel sorted set-enumeration tree named MIU-tree is designed to avoid the multiple database scans for discovering the set of HUIs. In addition, the utility-list structure is adopted to keep the necessary information for later mining process. Several efficient one-phase algorithms are then designed for handling the HUI-MMU problem.
	
	\item The designed \textit{global downward closure} (GDC) property and the \textit{conditional downward closure} (CDC) property guarantee that the global and partial anti-monotonicity can be maintained for mining the set of HUIs in the designed MIU-tree. With the two properties, HIMU can easily discover HUIs by pruning a huge number of unpromising itemsets.
	
	\item Based on the designed MIU-tree, three improved HIMU$_{EUCP}$, HIMU$_{LAP}$, and HIMU$_{ELP}$ algorithms are then further developed to reduce the search space by utilizing two efficient pruning strategies, which can speed up computations for discovering the set of HUIs.
	
	\item The extensive performance evaluation on synthetic and real-world datasets demonstrates the efficiency and effectiveness of the proposed algorithms in terms of runtime, number of generated patterns and scalability.
\end{itemize}

 The rest of this paper is organized as follows. Related work is briefly reviewed in Section \ref{sec:2}. Preliminaries and the problem statement of the addressed mining task are presented in Section \ref{sec:3}. The proposed HIMU algorithm and three improved variants  are respectively provided in Section \ref{sec:HIMU} and Section \ref{sec:HIMU2}. An experimental evaluation of the designed algorithms is provided in Section \ref{sec:6}. Finally, conclusion and future works are drawn in Section \ref{sec:conclusion}.

\section{Related Work}  \label{sec:2}

We structure the related work around the two main elements that this paper addresses: utility-oriented pattern mining and interesting pattern mining with multiple/individualized thresholds.

\subsection{Utility-oriented Pattern Mining}
% fournier2011rulegrowth

In real-world applications, the transactional data are always associated with purchase quantities, and every individual object represents its unique importance (such as interestingness, risk, and profit, etc.). Depending on different requirements in various domains and applications, the discovered knowledge can be generally classified as frequent itemsets or association rules \cite{agrawal1993mining,chen1996data,vo2013lattice,vo2015fast}, sequential patterns \cite{agrawal1995mining,srikant1996mining}, sequential rules \cite{fournier2012cmrules}, high-utility patterns \cite{chan2003mining,yao2004foundational,tseng2010up,ts},  and other interesting patterns \cite{geng2006interestingness,hong2009effective}. Among these, utility-oriented pattern mining (also called high-utility pattern mining) is a useful mechanism to understand economic consumer behavior with profitable products, and it has been considered as a critical issue and has been extensively discussed in recent decade. 

Yao et al. then defined a uniform framework called utility mining to discover the profitable itemsets while considering both the purchase quantity (internal utility) and their unit profit (external utility) of item/sets in transactional data \cite{yao2004foundational}. Since the traditional support-based \textit{downward closure} property of association-rule mining does not hold in HUIM, a huge number of candidates is necessary to be generated before obtaining the set of HUIs. To overcome this drawback, a transaction-weighted utilization (TWU) model was presented to keep the \textit{transaction-weighted downward closure} (TWDC) property \cite{liu2005two}, which can guarantee the correctness and completeness of the discovered HUIs. The incremental high-utility pattern tree (IHUP-tree) \cite{ahmed2009efficient} and high-utility pattern tree (HUP-tree) \cite{lin2011effective} were then developed to mine the set of HUIs by integrating the frequent pattern (FP)-tree structure \cite{han2004mining} and the TWU model. Tseng et al. further proposed the UP-Growth \cite{tseng2010up} and UP-Growth+ \cite{tseng2013efficient} algorithms to mine the set of HUIs using the designed UP-tree structure; both UP-Growth \cite{tseng2010up} and UP-Growth+ \cite{tseng2013efficient} have better results than the previous approaches. Lan et al. proposed an efficient projection-based indexing approach for mining HUIs and applied a pruning strategy to reduce the number of candidate itemsets \cite{lan2014efficient}. Recently, several algorithms that mine the set of HUIs using a single phase were respectively presented, such as the HUI-Miner \cite{liu2012mining}, d2HUP \cite{liu2012direct}, HUP-Miner \cite{krishnamoorthy2015pruning}, FHM \cite{fournier2014fhm}, EFIM \cite{zida2015efim}, and others \cite{song2016high,2gan2018survey}. These one-phase algorithms avoid the problem of generate-and-test mechanism to find the set of HUIs, and outperform the UP-Growth \cite{tseng2010up} and UP-Growth+ \cite{tseng2013efficient}; they mainly rely on the concept of remaining utility to reduce the search space and early prune the unpromising candidates. 

In addition to the above efficiency issues of utility mining, there are many effectiveness issues of utility mining have been extensively studied. Consider the dynamic data, incremental and decremental utility mining are introduced \cite{gan2018survey,2lin2015mining,wang2018incremental}. Consider other types of data, some algorithms of utility mining are developed to deal with uncertain data \cite{2lin2016efficient}, temporal data \cite{lin2015efficient}, and complex data containing negative values \cite{lin2016fhn,2gan2017mining}.  In recent years, the problem of high-utility sequential pattern mining attracts extensive attention and has been extensively studied, such as \cite{yin2012uspan,lan2014applying,alkan2015crom,wang2016efficiently}. The developments of other topics, e.g., top-$k$ \cite{tseng2016efficient}, correlated pattern \cite{lin2017fdhup,gan2017extracting}, discount strategies \cite{lin2016fast}, utility occupancy \cite{gan2018huopm}, privacy preserving \cite{gan2018privacy} for utility-oriented pattern mining are still in progress \cite{2gan2018survey}, but all of them are designed to discover high-utility patterns with a uniform \textit{minutil}. To the best of our knowledge, the HUI-MMU and HUI-MMU$_{TID}$ algorithms \cite{lin2016efficient} were the first work focusing on utility mining with multiple minimum utility thresholds.

\subsection{Frequent Pattern Mining with Multiple Minimum Supports}

In the literature, the ``\textit{rare item problem}"  was first identified while mining frequent patterns with a single \textit{minsup} threshold \cite{weiss2004mining}. It is as follows: i). If \textit{minsup} is set too high, we will miss those patterns that involve rare items. ii). In order to find the frequent patterns that involve both frequent and rare items, we have to set low \textit{minsup}. However, this may cause combinatorial explosion, producing too many frequent patterns.

In some real-world applications, some items may have hierarchical attributes, and their frequencies can be different. To address the ``\textit{rare item problem}" \cite{weiss2004mining} in frequent pattern mining, several approaches have been extensively studied, such as MSApriori \cite{liu1999mining}, CFP-Growth \cite{hu2006mining}, CFP-Growth++ \cite{kiran2011novel}, REMMAR \cite{liu2011discovering}, and FQSP-MMS \cite{huang2013discovery}, among others. The first work on the ``multiple minimum supports framework" was introduced by Liu et al. \cite{liu1999mining}, and an Apriori-like \cite{agrawal1994fast} approach named MSApriori was introduced. This framework allows users to specify multiple minimum support to reflect the natures of the items and their varied frequencies in the databases. Each item in MSApriori is associated with a \textit{minimum item support} (\textit{MIS}) value. The MSApriori algorithm suffers, however, from the same weakness as the Apriori-like approach; it generates a huge number of candidates and multiply re-scans the databases. Lee et al. then proposed a fuzzy mining algorithm to find useful fuzzy association rules with multiple minimum supports using a maximum constraint \cite{lee2004mining}. A conditional frequent pattern-growth (CFP-Growth) was then proposed \cite{hu2006mining} to mine frequent itemsets with multiple thresholds using the pattern-growth method based on an MIS-tree structure and \textit{MIN} (the minimum \textit{MIS} value of items in the database). However, CFP-Growth performs an exhaustive search on the constructed condition tree to discover the complete set of frequent patterns, which causes time-consuming problem. An enhanced CFP-Growth++ was then proposed to employ the \textit{least minimum support} (\textit{LMS}) instead of \textit{MIN} to reduce the search space and improve mining performance \cite{kiran2011novel}. LMS is the least value among all MIS values of frequent items. Recently, Gan et al. presented the state-of-the-art approach named FP-ME for pattern mining with multiple minimum supports from the Set-enumeration tree (ME-tree) \cite{gan2017mining}. 

Mining various types of patterns using multiple minimum supports has been extensively studied. In bio-informatics, the relational-based multiple minimum supports association rules (REMMAR) algorithm was proposed to discover relational association rules with multiple minimum supports in microarray datasets \cite{liu2011discovering}. Lee et al. also developed an approach to mine fuzzy multiple-level association rules with multiple minimum support thresholds \cite{lee2006mining}. The fuzzy quantitative sequential pattern with multiple minimum supports (FQSP-MMS) algorithm \cite{huang2013discovery} was proposed to discover fuzzy quantitative sequential patterns (FQSPs) by considering both multiple minimum supports and adjustable membership functions to handle quantitative databases based on fuzzy-set theory.

The previous studies on frequent pattern mining with multiple thresholds cannot, however, be directly applied to mine high-utility itemsets with multiple minimum utility thresholds since the unit profit and the quantity of the items are both considered. The HUI-MMU and HUI-MMU$_{TID}$ algorithms \cite{lin2016efficient} were first designed to mine HUIs with multiple minimum utility thresholds in the level-wise manner, but they may cause the combinational problem and memory leakage, which are common problems of the Apriori-like mechanism \cite{liu2005two}. It is thus a challenging issue to efficiently and effectively mine HUIs with multiple minimum utility thresholds and will be discussed and studied as follows.

\section{Preliminaries and Problem Statement} \label{sec:3}

According to the utility theory \cite{marshall1926principles} and some definitions of previous studies \cite{yao2004foundational,tseng2010up,lin2016efficient}, we have the following concepts and formulation.

Let $I$ = $\{i_1, i_2, \dots, i_m\}$ be a finite set of $m$ distinct items appearing in a quantitative  transaction database $D$ = $\{T_1, T_2, \dots, T_n\}$, where each quantitative transaction $T_{q} \in D$ (1 $ \leq q \leq n $) is a subset of $I$, and has a unique identifier called its \textit{TID}. Each product/item $i_j$ in a transaction has a purchase quantity (also called \textit{internal utility}) denoted $q(i_j, T_q)$. In addition, each product/item $i_j$ (1 $ \leq j \leq m $) in $D$ has its unique user-specified profit value $pr(i_m)$ (also called \textit{external utility}, which represents its importance, e.g., profit, interest, risk.). Unit profits are stored in a profit-table \textit{ptable} = \{$ pr(i_{1})$, $ pr(i_{2})$, $\dots$, $ pr(i_{m}) $\}. 
An itemset \textit{X}$\subseteq$\textit{I} with \textit{k} distinct items \{$ i_{1}$, $ i_{2}$, $\dots$, $ i_{k} $\} is of length \textit{k} and is referred to as a \textit{k}-itemset. An itemset \textit{X} is said to be contained in a transaction $ T_{q} $ if $ X\subseteq T_{q}$. For an itemset $X$, let the notation \textit{TIDs(X)} denotes the \textit{TIDs} of transactions in $D$ containing $X$.

A quantitative database in Table \ref{table_example2} is given as an example to illustrate the process of the proposed algorithm. It consists of 10 transactions and five products/items, denoted from $a$ to $e$. Note that Table \ref{profit_table} shows the profit table, which defines the unit profit value (\$) of each product/item.

\begin{table}[!htbp]
	\centering
	\small
	\caption{An example database.} 
	\label{table_example2}
	\begin{tabular}{|c|c|c|c|}
		\hline
		\textbf{TID} & \textbf{Time} & \textbf{Customer ID}  & \textbf{Event (product: quantity)} \\ \hline
		$T_1$   &   05-08-2017 10:45:30   &   $C_1$  &  \textit{a}:1, \textit{c}:2, \textit{d}:3 	    \\ 
		$T_2$   &   05-08-2017 10:59:12   &   $C_2$  &  \textit{a}:2, \textit{d}:1, \textit{e}:2	    \\ 
		$T_3$   &   05-08-2017 11:05:40   &   $C_3$  &  \textit{b}:3, \textit{c}:5	    \\ 
		$T_4$   &   05-08-2017 11:40:00   &   $C_4$  &  \textit{a}:1, \textit{c}:3, \textit{d}:1, \textit{e}:2	    \\ 
		$T_5$   &   05-08-2017 12:55:14   &   $C_3$  &  \textit{b}:1, \textit{d}:3, \textit{e}:2	    \\ 
		$T_6$   &   05-08-2017 14:08:58   &   $C_2$  &  \textit{b}:2, \textit{d}:2  \\ 
		$T_7$   &   05-08-2017 14:40:00   &   $C_5$  &  \textit{b}:3, \textit{c}:2, \textit{d}:1, \textit{e}:1	    \\ 
		$T_8$   &   05-08-2017 15:01:40   &   $C_4$  &  \textit{a}:2, \textit{c}:3	    \\ 
		$T_9$   &   05-08-2017 15:04:26   &   $C_6$  &  \textit{c}:2, \textit{d}:2, \textit{e}:1    \\ 		
		$T_{10}$ &  05-08-2017 15:30:20   &   $C_7$  &  \textit{a}:2, \textit{c}:2, \textit{d}:1	    \\ 						
		\hline
	\end{tabular}
\end{table}

\begin{table}[!htbp]
	\centering
	\small
	\caption{A profit table.} 
	\label{profit_table}
	\begin{tabular}{|c|c|c|}
	\hline
	\textbf{Product} & \textbf{Profit(\$)} \\ \hline
	$ a $ & 6  \\ \hline
	$ b $ & 12 \\ \hline
	$ c $ &	1  \\ \hline
	$ d $ &	9  \\ \hline
	$ e $ &	3  \\ \hline	
\end{tabular}
\end{table}

\begin{definition}
	\rm The minimum utility threshold of an item $ i_{j} $ in a database \textit{D} is denoted as $ mu(i_{j})$. A structure called \textit{MMU-table} indicates the minimum utility thresholds of each item in \textit{D}, and defined as:
	\begin{equation}
	   \textit{MMU-table} = \{mu(i_{1}), mu(i_{2}), \dots, mu(i_{m})\}.
	\end{equation}
\end{definition}

Assume that the minimum utility thresholds of the five items in the running example are defined and shown in Table \ref{MMU_table}, such that: \textit{MMU-table} = $ \{mu(a), mu(b), mu(c), mu(d), mu(e)\}$ = \{\$56, \$65, \$53, \$50, \$70\}. To avoid the ``\textit{rare item problem}'', we consider the smallest utility threshold among items in an itemset as its minimum utility threshold, as defined below.

\begin{table}[!htbp]
	\centering
	\small
	\caption{A predefined \textit{MMU-table}.} 
	\label{MMU_table}
	\begin{tabular}{|c|c|c|}
		\hline
		\textbf{Product} & \textbf{Minimum utility (\$)} \\ \hline
		$ a $ & 56  \\ \hline
		$ b $ & 65 \\ \hline
		$ c $ &	53  \\ \hline
		$ d $ &	50  \\ \hline
		$ e $ &	70  \\ \hline	
	\end{tabular}
\end{table}

\begin{definition}
	\rm The minimum utility threshold of a \textit{k}-itemset \textit{X} = \{$ i_{1}$, $ i_{2}$, $\dots$, $ i_{k} $\} in \textit{D} is denoted as \textit{MIU}$(X)$, and defined as the smallest \textit{mu} value for items in \textit{X}, that is:
	\begin{equation}
	     MIU(X) = min\{mu(i_{j})|i_{j}\in X, 1\leq j \leq  k\}.
	\end{equation}
\end{definition}

For example, \textit{MIU}$(a)$ = $ min\{mu(a)\} $ = 56 \$, \textit{MIU}$(ac)$ = $ min\{mu(a),  mu(c)\}$ =  $min$\{\$56, \$53\} = \$53, and \textit{MIU}$(ace)$ =  $ min \{mu(a), mu(c), mu(e)\} $ = $min$\{\$56, \$53, \$70\} = \$53.

\begin{definition}
	\rm The utility of an item $ i_{j} $ in a transaction $ T_{q} $ is denoted as $u(i_{j}, T_{q})$ and defined as:
	\begin{equation}
	    u(i_{j}, T_{q}) = q(i_{j}, T_{q})\times pr(i_{j}).
	\end{equation}
\end{definition}

For example, the utility of item $(b)$ in transaction $T_3$ is calculated as $u(b, T_3)$ = $q(b, T_3) \times pr(b)$ = 3 $\times $ \$12 = \$36.

\begin{definition}
	\rm The utility of an itemset \textit{X} in a database \textit{D} is denoted as $u(X)$ and defined as:
	\begin{equation}
	    u(X) = \sum_{X\subseteq T_{q}\wedge T_{q}\in D} u(X, T_{q}),
	\end{equation}
	in which $u(X, T_{q})$ is the utility of an itemset \textit{X} in a transaction $ T_{q} $, and defined as:
	\begin{equation}
	u(X, T_{q}) = \sum _{i_{j}\in X\wedge X\subseteq T_{q}}u(i_{j}, T_{q}).
	\end{equation}
\end{definition}

For example, the utility of the itemset $(bc)$ in transaction $T_3$ is calculated as $u(bc, T_3)$ = $u(b, T_3)$ + $u(c, T_3)$ = $q(b, T_3) \times pr(b)$ + $q(c, T_3) \times pr(c) $ = 3 $\times $ \$12 + 5 $\times $ \$1 = \$41. Therefore, the utility of itemsets $(b)$ and $(bc)$ in $D$ are respectively calculated as $u(b)$ = $u(b, T_3)$ + $u(b, T_5)$ + $u(b, T_6)$ + $u(b, T_7)$ = \$36 + \$12 + \$24 + \$36 = \$108, and $u(bc)$ = $u(bc, T_3)$ + $u(bc, T_7)$ = \$41 + \$38 = \$79.

\begin{definition}
	\rm The transaction utility of a transaction $ T_{q} $  is denoted $tu(T_q)$ and defined as:
	\begin{equation}
     	tu(T_{q}) = \sum_{i_{j}\in T_{q}}u(i_{j}, T_{q}).
	\end{equation}
	in which $j$ is the number of items in $T_{q}$.
\end{definition}

For example, $tu(T_5)$ = $u(b, T_5)$ + $u(d, T_5)$ + $u(e, T_5)$ = \$12 + \$27 + \$6 = \$45. Therefore, the transaction utilities of $T_1$ to $T_{10}$ are respectively calculated as $tu(T_1)$ = \$35, $tu(T_2)$ = \$27, $tu(T_3)$ = \$4, $tu(T_4)$ = \$24, $tu(T_5)$ = \$45, $tu(T_6)$ = \$42, $tu(T_7)$ = \$50, $tu(T_8)$ = \$15, $tu(T_9)$ = \$23, and $tu(T_{10})$ = \$23. 

\begin{definition}
	\rm The transaction-weighted utility of an itemset $ X $ is denoted as $ TWU(X)$, which is the sum of all transaction utilities containing $ X $ and defined as:
	\begin{equation}
	    TWU(X)=\sum_{X\subseteq T_{q}\wedge T_{q}\in D}tu(T_{q}).
	\end{equation}
\end{definition}

\begin{definition}
	\rm An itemset $ X\subseteq I $ is said to be a high transaction-weighted utilization itemset (HTWUI) in a database if its TWU value is no less than the minimum utility threshold \cite{lin2016efficient}. To adapt this definition for HUI-MMU, we assume that this threshold is the \textit{MIU}$(X)$ which has been defined in Definition 3.2. Thus: 
	\begin{equation}
	    HTWUI\leftarrow\{X|TWU(X)\geq MIU(X)\}.
	\end{equation}
\end{definition}

For example, $TWU(b)$ = $tu(T_3)$ + $tu(T_5)$ + $tu(T_6)$ + $tu(T_7)$ = \$41 + \$45 + \$42 + \$50 = \$178. Because $TWU(b) \geq MIU(b)$, the item ($b$) is an HTWUI.

\begin{definition}
	\rm An itemset \textit{X} in a database \textit{D} is a high-utility itemset (HUI) if and only if its utility is no less than its minimum utility threshold:
	\begin{equation}
	    HUI\leftarrow\{X|u(X)\geq MIU(X)\}.
	\end{equation}
\end{definition}

Assume that the multiple minimum utility thresholds of items are shown in Table \ref{MMU_table}. The item $(a)$ is not an HUI since its utility is calculated as $u(a)$ = \$48, which is less than $MIU(a)$ = \$56. However, the itemset $(ad)$ is an HUI since its utility is calculated as $u(ad)$ = \$90, which is greater than its minimum utility value $MIU(ad)$ = $min\{mu(a), mu(d)\}$ = $min$\{\$56, \$5\} = \$50. For the running example, the complete set of HUIs with multiple minimum utility thresholds is shown in Table \ref{table_HUIs}.

\begin{table}[ht]	
	\caption{Derived HUIs}
	\centering
	\label{table_HUIs}
	\begin{tabular}{|c|c|c|c|c|c|}
		\hline
		\textbf{Itemset} & \textbf{\textit{MIU}}	 & \textbf{Utility} &	\textbf{Itemset} &	\textbf{\textit{MIU}}	& \textbf{Utility} \\ \hline	
		$ (b) $ &	\$65  &	\$108  &   $ (de) $ &	\$50  &	\$96    \\ \hline
		$ (d) $ &	\$50  &	\$126 	&	$ (acd) $ &	\$50  &	\$76 	\\ \hline	
		$ (ad) $ &	\$50  &	\$90   &	$ (bde) $ &	\$50  & \$93    \\ \hline
		$ (bc) $ &	\$53  &	\$79   &   $ (cde) $ &	\$50  &	\$55    \\ \hline
		$ (bd) $ &	\$50  &	\$126  &	$ (bcde)$ & \$50  &	\$50    \\ \hline	
		$ (cd) $ &	\$50  &	\$83   &	& &  \\ \hline 	
	\end{tabular}
\end{table}

\textbf{Problem Definition:} Given a transactional database $D$ containing $m$ items, the user-specified unit minimum utility threshold $\{mu(i_{1}), mu(i_{2}), \dots, mu(i_{m})\}$ of each item is formed as an \textit{MMU-table}. The problem of high-utility itemset mining with multiple minimum utility thresholds (HUIM-MMU) is to find the complete set of the $k$-itemsets in which the utility of each itemset $X$ is greater than \textit{MIU(X)}.

Hence, the proposed model facilitates the user to specify each item with varied item-specific minimum utility threshold. 
As a result, this HUIM-MMU model can address the ``\textit{rare item problem}" and efficiently find the complete set of HUIs in the database while considering multiple minimum utility thresholds rather than considering only a single minimum utility threshold. Thus, the problem of HUIM-MMU is quite different and complicated than the traditional HUIM problem.

\section{Proposed HIMU Algorithm for High-Utility Itemset Mining} \label{sec:HIMU}

In the past, a generate-and-test algorithm named HUI-MMU was first proposed to mine the set of HUIs with multiple minimum utility thresholds \cite{lin2016efficient}. A two-phase upper-bound model was proposed to keep the \textit{sorted downward closure} (\textit{SDC}) property while mining the set of HUIs under MMU. Although the upper-bound strategy was proposed in HUI-MMU, and the enhanced HUI-MMU$_{TID}$ algorithm with TID-index was further designed to early prune the huge number of unpromising candidates. However, a drawback of those two approaches is that the level-wise search mechanism may easily suffer from the ``exponential explosion" problem especially when the MMUs are set quite low. Thus, a novel one-phase algorithm named HIMU is designed, which can directly mine the set of HUIs without candidate generation by spanning the developed MIU-tree with the \textit{GDC properties}. Details are presented in the following section.

\subsection{Search Space of HIMU}
The search space of the addressed HUIM-MMU problem can be represented as a lattice structure \cite{pasquier1998pruning} or a set-enumeration tree \cite{rymon1992search}. The complete search space of HUIM-MMU with $n$ distinct items is calculated as $2^n$-1. From Table \ref{table_HUIs}, it can be observed that the well-known \textit{downward closure} property of association-rule mining and the \textit{transaction-weighted downward closure} (TWDC) of traditional high-utility itemset mining property are inappropriate for the addressed HUIM-MMU. For example, an item $(c)$ is not an HUI, but its supersets $(bc)$, $(cd)$, $(acd)$, $(cde)$, and $(bcde)$ may consider as the HUIs. Without the anti-monotone constraint, a huge number of candidates would be generated to obtain the final HUIs. Based on the TWDC property, the number of candidates in traditional HUIM can be considerably reduced \cite{liu2005two}. However, it can be easily seen that the TWDC property does not hold for the addressed HUIM-MMU. For example, considering items $(b)$, $(c)$, $(d)$, and $(e)$, their minimum utility thresholds respectively are $MIU(b)$ = \$65, $MIU(c)$ = \$53, $MIU(d)$ = \$50, and $MIU(e)$ = \$70. The TWU value of the itemset $(bce)$ is calculated as $TWU(bce)$ = \$50, which is less than the minimum high-utility values of all its subsets as $MIU(a)$, $MIU(e)$, and $MIU(f)$; the itemset $(bce)$ is not an HTWUI. In this situation, the itemset $(bcde)$ and its supersets would be ignored according to the TWDC property. It can be observed, however, that the $TWU(bcde)$ = \$50, which is equal to $MIU(bcde)$ = \$50. 

Therefore, the TWU measure does not maintain the TWDC property of HUIM-MMU. From the given results shown in Table \ref{table_HUIs}, it can be seen that itemset $(bcde)$ is actually an HUI. It is incorrect to discard the supersets of $(bce)$ based on the TWDC property, since $(bcde)$ would not be generated in this situation. Thus, the TWDC property does not suitable for HUIM-MMU. 

\textbf{Observation 1.} In traditional HUIM, the TWDC property ensures that the set of HTWUIs contains all HUIs. However, this property does not hold in the HUIM-MMU framework, which can be proven as follows.

\begin{proof}
Let $X^{k-1}$ be a ($k$-1)-itemset and $X^{k-1}$ and $X^k$ = $\{i_1, i_2, \dots, i_k\}$ be any superset of length $k$ of $X^{k-1}$. Because $X^{k-1} \subseteq X^k$, the two following relationships hold:

(1) By definition of \textit{MIU} concept, we have that $MIU(X^{k-1})$ = $min\{mu(i_1), mu(i_2), \dots, mu(i^{k-1})\}$ and $MIU(X^k)$ = $min\{mu(i_1)$, $mu(i_2)$, $ \dots, mu(i_k)\}$. From Definition 1, it can be found that $MIU(X^k) \leq MIU(X^{k-1})$.

(2) Thus, $TWU(X^{k})=\sum_{X^{k}\subseteq T_{q}\wedge T_{q}\in D}tu(T_{q})\leq\sum_{X^{k-1}\subseteq T_{q}\wedge T_{q}\in D}tu(T_{q})=$ $TWU(X^{k-1})$.
\end{proof}

Therefore, if $X^k$ is an HTWUI [i.e., $TWU(X^k) \geq MIU(X^k)$], we cannot ensure that any subset $X^{k-1}$ of $X^k$ is also an HTWUI since $MIU(X^k) \leq MIU(X^{k-1})$. The TWDC property of the traditional TWU model does not hold in the HUIM-MMU framework.

To address this limitation, a new \textit{sorted downward closure} (SDC) property was proposed in the HUI-MMU algorithm \cite{lin2016efficient}. However, HUI-MMU algorithm applies the level-wise mechanism to exploit the set of HUIS, which is time-consuming and not suitable to efficiently mine the set of HUIs. It is thus a critical issue to design a suitable data structure and efficient pruning strategies to efficiently reduce the search space and early prune the unpromising itemsets for mining the set of HUIs with varied item-specific minimum utility thresholds.

\subsection{Proposed MIU-tree and the Sorted Downward Closure Property}

\begin{definition}[Total order $\prec$ on items]
	\rm  Assume that the items in each transaction in a database are sorted according to their lexicographic order. We also assume that the total order $\prec$ on items is the ascending order of minimum utility thresholds of 1-itemsets.
\end{definition}

For example in Table \ref{table_example2}, the ascending order of minimum utility thresholds of 1-itemsets in the running example is $MIU(d) < MIU(c) < MIU(a) < MIU(b) < MIU(e)$; the total order $\prec$ on items is thus as: $ d \prec c \prec a \prec b \prec e $.

\begin{definition}[Set-enumeration tree with multiple minimum item utilities, MIU-tree]
	\rm The designed MIU-tree structure is a sorted set-enumeration tree where the total order $\prec$ on items is the ascending order of minimum utility thresholds of items.
\end{definition}

\begin{definition}
	\rm  The extensions (descendant nodes) of an itemset (tree node) $X$ can be obtained by appending an item $y$ to $X$ such that $y$ is greater than all items already in $X$ according to the total order $\prec$.
\end{definition}

For example, an illustrated MIU-tree of the HIMU algorithm of the running example is shown in Fig. \ref{MIU_tree}. For example, a node $(dc)$ is used as an illustrated tree node shown in Fig. \ref{MIU_tree} . The 1-extensions of node $(dc)$ are the nodes $(dca)$, $(dcb)$, and $(dce)$, and the all extension nodes (w.r.t. all its descendant nodes) of $(dc)$ are the nodes $(dca)$, $(dcb)$, $(dce)$, $(dcab)$, $(dcae)$, $(dcbe)$, and $(dcabe)$. Based on the constructed MIU-tree, the following lemmas are obtained.

%lemma 1
\begin{lemma}
	\label{lemma1}
	The complete search space of the HUIM-MMU problem can be represented as an MIU-tree, where the ascending order of MIU of items is adopted.
\end{lemma}

\begin{proof} 
	According to the definition of a set-enumeration tree \cite{rymon1992search}, the complete search space of $I$ (where $m$ is the number of items in $I$) contains $2^m$-1 patterns if all subsets of I are enumerated. For example, the MIU-tree shown in Fig. \ref{MIU_tree} illustrates that all subsets of $I$ = $\{a, b, c, d, e\}$, w.r.t. all the possible patterns. Thus, all the supersets of the root node can be enumerated according to any order of 1-itemsets (e.g., the designed total order $\prec$, the lexicographic order, the descending order of support/utilities, etc.). The developed MIU-tree can thus be used to represent the search space of the designed approach.
\end{proof}

\begin{figure}[hbtp]	
	\centering
	\includegraphics[trim= 0 0 5 0,clip,scale=0.2]{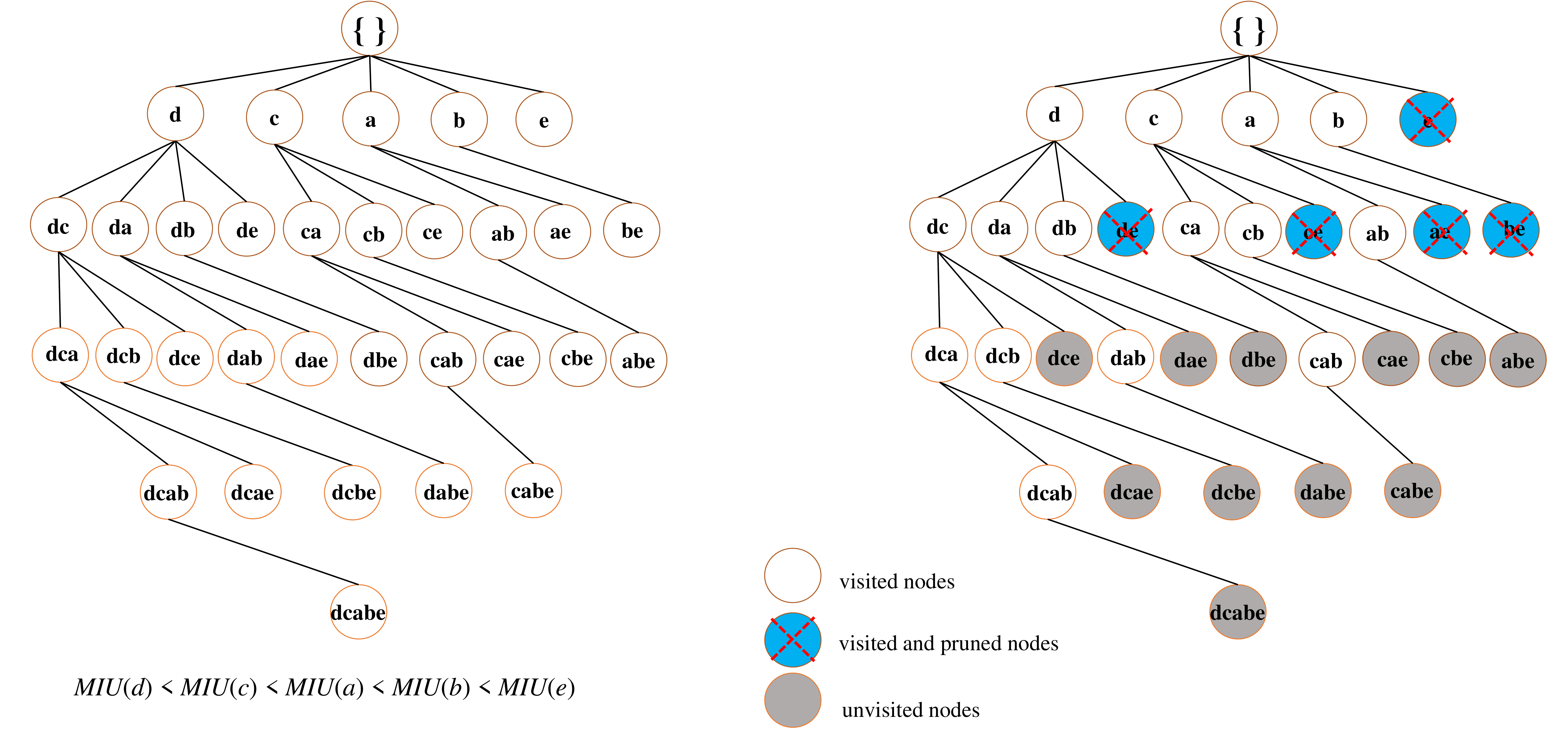}
	\captionsetup{justification=centering}
	\caption{The MIU-tree representation of the search space.}
	\label{MIU_tree}
\end{figure}

The traditional TWDC property of the TWU model does not hold in the proposed HUIM-MMU framework. For example, consider items $(b)$, $(c)$, $(d)$ and $(e)$, the results of \textit{MIU} values of them are (\textit{MIU}$(b)$: \$65, \textit{MIU}$(c)$: \$53, \textit{MIU}$(d)$: \$50 and \textit{MIU}$(e)$: \$70). The TWU of an itemset $ (bce) $ is calculated as $TWU(bce)$ = \$50, which is less than the minimum utility values of its subsets \textit{MIU}$(b)$, \textit{MIU}$(c)$ and \textit{MIU}$(e)$. Hence, $ (bce) $ is not a HTWUI, and thus the itemset $ (bcde) $ and its supersets would be discarded according to the TWDC property. But it can be observed that $ TWU(bcde) $ = \$50, which is equal to $ MIU(bcde) $ = \$50. But as shown in Table \ref{table_HUIs}, it can be seen that itemset $ (bcde) $ is actually an HUI. 

It is thus incorrect to discard the supersets of $ (bce) $ based on the TWDC property since $(bcde)$ would not be generated. Therefore, if a $k$-itemset $ X^{k} $ is a HTWUI (i.e., $ TWU(X^{k})\geq MIU(X^{k}) $), we cannot ensure that any subset $ X^{k-1} $ of $ X^{k} $ is also a HTWUI (because $ MIU(X^{k-1})\geq MIU(X^{k}) $). Thus, using this property to prune the search space may fail to discover the complete set of HUIs. To address this limitation, the \textit{Sorted Downward Closure (SDC) property} was proposed in \cite{lin2016efficient}. 

\textbf{Theorem 1.} (\textit{Sorted downward closure property}, \textit{SDC property}) \emph{Assume that items in itemsets are sorted by ascending order of mu values. Given any itemset $ X^{k} $ = \{$ i_{1}$, $ i_{2}$, $\dots$, $ i_{k} $\} of length \textit{k}, and another itemset $ X^{k-1} $ = \{$ i_{1}$, $ i_{2}$, $\dots$, $ i_{k-1} $\}  such that $ X^{k-1}\subseteq X^{k}$. If $ X^{k} $ is a HTWUI then $ X^{k-1} $ is also a HTWUI} \cite{lin2016efficient}.

\begin{proof}
	Since $ X^{k-1}\subseteq X^{k}$, the following relationships hold:\\
	
	(1) By definition 2, we have that $ MIU(X^{k-1}) $ = $ min $\{$ mu(i_{1})$, $ mu(i_2)$, $\dots$, $ mu(i_{k-1}) $\}, and $ MIU(X^{k}) $ = $ min $\{$ mu(i_{1})$, $ mu(i_{2})$, $\dots$, $ mu(i_{k}) $\}. Since \{$ i_{1}$, $ i_{2}$, $\dots$, $ i_{k} $\} is sorted according to the total order $\prec$, $ MIU(X^{k}) $ = $ MIU(X^{k-1}) $ = $ mu(i_{1})$. \\
	
	(2) Thus, $TWU(X^{k})=\sum_{X^{k}\subseteq T_{q}\wedge T_{q}\in D}tu(T_{q})\leq\sum_{X^{k-1}\subseteq T_{q}\wedge T_{q}\in D}tu(T_{q})=$ $TWU(X^{k-1})$.
	
	Therefore, if $ X^{k} $ is a HTWUI (i.e., $ TWU(X^{k})\geq mu(i_{1}) $), any subset $ X^{k-1} $ of $ X^{k} $ is also a HTWUI. 
\end{proof}

Although the \textit{SDC} property can guarantee the anti-monotonicity of HTWUIs, some HUIs may be discarded if the 1-itemsets of HTWUI were first determined by the \textit{MIU}($X$) values. For example, assume that an itemset $X$ has its TWU value and is less than $MIU(X)$, $X$ is not regarded as an HTWUI, and an itemset $Y$ as any extension of $X$ is not considered as an HTWUI either. However, if $Y$ = $\{i_j \cup X \}$, in which ij has a lower mu value than $X$, $Y$ still would be an HUI when $u(Y) \geq MIU(i_j)$. If the \textit{SDC} property is directly applied to produce the 1-itemsets of HTWUIs, $Y$ will not be considered an HTWUI since $X$ is not an HTWUI (i.e., $TWU(X) < MIU(X)$). It is thus inappropriate to directly determine the 1-itemsets of HTWUIs based on the \textit{SDC} property. To address this problem, the concept of least minimum utility (\textit{LMU}) was developed to guarantee the completeness of the derived HUIs with multiple minimum utility thresholds \cite{lin2016efficient}.

\begin{definition}[Least minimum utility value, \textit{LMU}]
	The least minimum utility value (\textit{LMU}) is defined as the smallest value in the \textit{MMU-table}, that is:
	\begin{equation}
	     LMU = min\{mu(i_{1}), mu(i_{2}), \dots, mu(i_{m})\}, 
	\end{equation}
	where \textit{m} is the total number of items in a database.
\end{definition}

For example, the \textit{LMU} of the given example is calculated as: $ min$\{$ mu(a)$, $ mu(b)$, $ mu(c)$, $ mu(d)$, $ mu(e) $\} = $ min $\{\$56, \$65, \$53, \$50, \$70\} = \$50.

\subsection{Conditional Downward Closure Property and Global Downward Closure Property for HIMU}

Based on the constructed MIU-tree, the following lemmas can be obtained.

%lemma 2
\begin{lemma}
	\label{lemma2}
	The \textit{MIU} value of a node/pattern in the MIU-tree is equal to that of any of its child nodes (extension nodes).
\end{lemma}

\begin{proof}
	Assume that $X^{k-1}$  is a node representing an itemset $X$ in the MIU-tree, and that  $X^{k}$ is  any of its child nodes (extensions). We have that \textit{MIU}$(X^{k-1})$ = $min\{mu(i_1)$, $ mu(i_2), \dots, mu(i_{k-1)}\}$, and \textit{MIU}$(X^k)$ = $min\{mu(i_1),$ $mu(i_2), \dots, mu(i_k)\}$. Since  $\{i_1, i_2, \dots, i_k\} $ is sorted by ascending order of $mu$ values, it can be proven that:  $MIU(X^k) = MIU(X^{k-1}) = mu(i_1)$. Thus, the $MIU$ value of a node in the MIU-tree is always equal to the \textit{MIU} of any of its child nodes. 
\end{proof}

For example, in Fig. \ref{MIU_tree}, the MIU-tree is sorted in MIU-ascending order of 1-itemsets. The MIU of $(dc)$ is equal to any of its extension nodes such as $(dca)$, $(dcb)$, $(dce)$, $(dcab)$, $(dcae)$, and $(dcabe)$, which can be calculated as $MIU(dc)$ = $min$\{\$50, \$53\} = \$50. From these results, it can be found that the MIU of all the extension nodes of $(dc)$ are all the same as $MIU(dc)$. According to MIU-tree, we have the following lemma.

%lemma 2
\begin{lemma}
	\label{lemma3}
	The support of a node in the MIU-tree is no less than the support of any of its child nodes (extension nodes).
\end{lemma}

\begin{proof}
	Since the Set-enumeration MIU-tree is a shared prefix tree, the relationship of the support of $X^{k}$ and $X^{k-1}$ can be proven to be $sup(X^k) \leq sup(X^{k-1})$. 
\end{proof}

\textbf{Theorem 2.} ($\mathbf{HUIs\subseteq HTWUIs}$)
  \emph{Assume that 1-itemsets having a TWU lower than $LMU$ are discarded and that the total order $\prec$ is applied. We have that HUIs $\subseteq$ HTWUIs, which indicates that if an itemset is not a HTWUI, then it is not an HUI. Moreover, none of its extensions are HTWUIs or HUIs.}

\begin{proof}
	Let $X^{k}$ be an itemset such that $ X^{k-1} $ is a subset of $ X^{k}$.\\ 
	(1)	We have that $ TWU(X^{1})\leq LMU $ and \textit{MIU}$(X^{k})\geq LMU$.\\
	(2)	Since items are sorted by ascending order of \textit{mu} values, $ TWU(X^{k-1})\geq TWU(X^{k})$ and \textit{MIU}$(X^{k-1}) = MIU(X^{k}) $ = $ min $\{$ mu(i_{1})$, $ mu(i_{2})$, $\dots$, $ mu(i_{m}) $\} = $ mu(i_{1})$.\\
	(3)	$u(X)=\sum_{X\subseteq T_{q}\wedge T_{q}\in D}u(X, T_{q})\leq\sum_{X\subseteq T_{q}\wedge T_{q}\in D}tu(T_{q}) = TWU(X)$.\\
	Thus, if $ X^{k-1} $ is not a HTWUI and $ TWU(X^{k-1}) < mu(i_{1})$, none of its supersets are HUIs. 
\end{proof}

%lemma 4
\begin{lemma}
	\label{lemma4}
	The TWU of any node in the Set-enumeration MIU-tree is no less than the sum of all the actual utilities of any one of its descendant nodes, but not the \textit{MIU} of its descendant nodes.
\end{lemma}

\begin{proof}
	Let $X^{k-1}$ be a node in the MIU-tree, and $X^{k}$ be a children (extension) of $X^{k-1}$. According to Theorem 1 and Lemma 1, we can get $TWU(X^{k-1})$ $\geq$ $TWU(X^k)$ and the relationship between \textit{MIU} values. Thus, the lemma holds. 
\end{proof}

\textbf{Theorem 3.} (\textit{Global downward closure property}, \textit{GDC property})	\label{theorem3} 
	\emph{In the designed MIU-tree, if the TWU value of a tree node $X$ is less than the \textit{LMU}, $X$ is not an HUI, and all its supersets (not only all its descendant nodes, but also the other nodes containing $X$) are also not considered as HUI.}

\begin{proof}
	According to Lemma 2 and Theorem 2, this theorem holds. 
\end{proof}	

This theorem ensures that by discarding  itemsets with a TWU less than $LMU$, and their extensions, no HUIs are missed. Thus, the designed \textit{global downward closure (GDC) property} and the \textit{LMU} guarantee the \textbf{completeness} and \textbf{correctness} of the proposed HIMU algorithm, when pruning the search space.

\textbf{Remark 1.} The TWU of any node in the set-enumeration MIU-tree is greater than or equal to the sum of all the actual utilities of its descendant nodes, but the MIU of descendant nodes.

\begin{proof}
   Let $X^{k-1}$ be a node in the MIU-tree, and $X^k$ be any of the children (extension nodes) of $X^{k-1}$. According to Theorem 1, we can obtain $TWU(X^{k-1}) \geq TWU(X^k)$. Thus, the property holds.
\end{proof}

In the past, a structure named utility-list was proposed in the HUI-Miner algorithm to directly mine HUIs without candidate generation \cite{liu2012mining}. The utility-list structure is efficient, and we adopt this structure in the proposed HIMU algorithm to store the necessary information. The utility-list structure keeps information from transactions in memory. Each entry in the utility-list of an itemset $X$ represents a transaction $T_q$ containing $X$. It consists of the corresponding TID (denoted as \textit{tid}), the utility of $X$ in $T_q$ (denoted as \textit{iu}), and the remaining utility of $X$ in $T_q$ (denoted as \textit{ru}). For two item/sets $X$ and $Y$, the notation $X \prec Y$ indicates that $X$ precedes $Y$ according to the total order on items (e.g., alphabetical order or TWU ascending order). As mentioned, the extension of an itemset $X$ is that $\{X \cup i\}$, where $i$ is an item.

\begin{definition}
 \rm  Give an itemset $X$ and a transaction (or itemset) $T$ such that $X \subseteq T,$ the set of all items from T that are not in $X$ is denoted as  $T \setminus X$, and the set of all the items appearing after $X$ in $T$ is denoted $T/X$. Thus, $T/X \subseteq T \setminus X$.
\end{definition}

For example, consider $X$ = $\{cd\}$ and the transaction $T_4$ in Table \ref{table_example2}; then, $T_4 \setminus X$ = $\{ae\}$ and $T_4 / X$ = $\{e\}$.

\begin{definition}[utility-list \cite{liu2012mining}]
	\rm  An entry in the utility-list of an itemset $X$ represents a transaction $T_q$ with three fields: (1) the TID of $X$ in $T_q$ (denoted as \textit{tid}, $X \subseteq T_q  \cup T_{q}\in D),$ (2) the utility of $X$ in $T_q$ (denoted as \textit{iu}), and (3) the remaining utility of $X$ in $T_q$ (denoted as \textit{ru}), where \textit{iu} is defined as $ X.iu = \sum_{i \in X \wedge X\subseteq T_{q}}u(i,T_{q}) $ of each transaction $T_q$, and \textit{ru} is defined as $ X.ru = \sum_{i \notin X \wedge i\subseteq T_{q} \wedge X \prec i}u(i,T_{q}) $ of each transaction $T_q$ \cite{liu2012mining}.
\end{definition}

Details of the construction of utility-list structures can be found in reference \cite{liu2012mining}. The utility-list of a $k$-itemset ($k > 1$) is constructed using the utility-list structures of its subsets of length $k$-1. Note that it is necessary to initially construct the utility-list structures of 1-itemsets with \textit{LMU} instead of HTWUI1 to recursively obtain the utility-list structures of larger itemsets based on Theorem 2. In the running example, \textit{LMU} = \$50 and the TWU value of 1-itemsets are \{$a$: \$124; $b$: \$178; $c$: \$211; $d$: \$269; $e$: \$169\}, where \{$a$: \$124\} indicates that the itemset $\{a\}$ has a TWU value of \$124. The utility-list structures are constructed in ascending order of their \textit{MIU} values ($d \prec c \prec a \prec b \prec e$), as shown in Fig. \ref{fig:Fig1}.

\begin{figure*}[hbtp]	
	\centering
	\includegraphics[trim = 0 0 5 0,clip,scale=0.45]{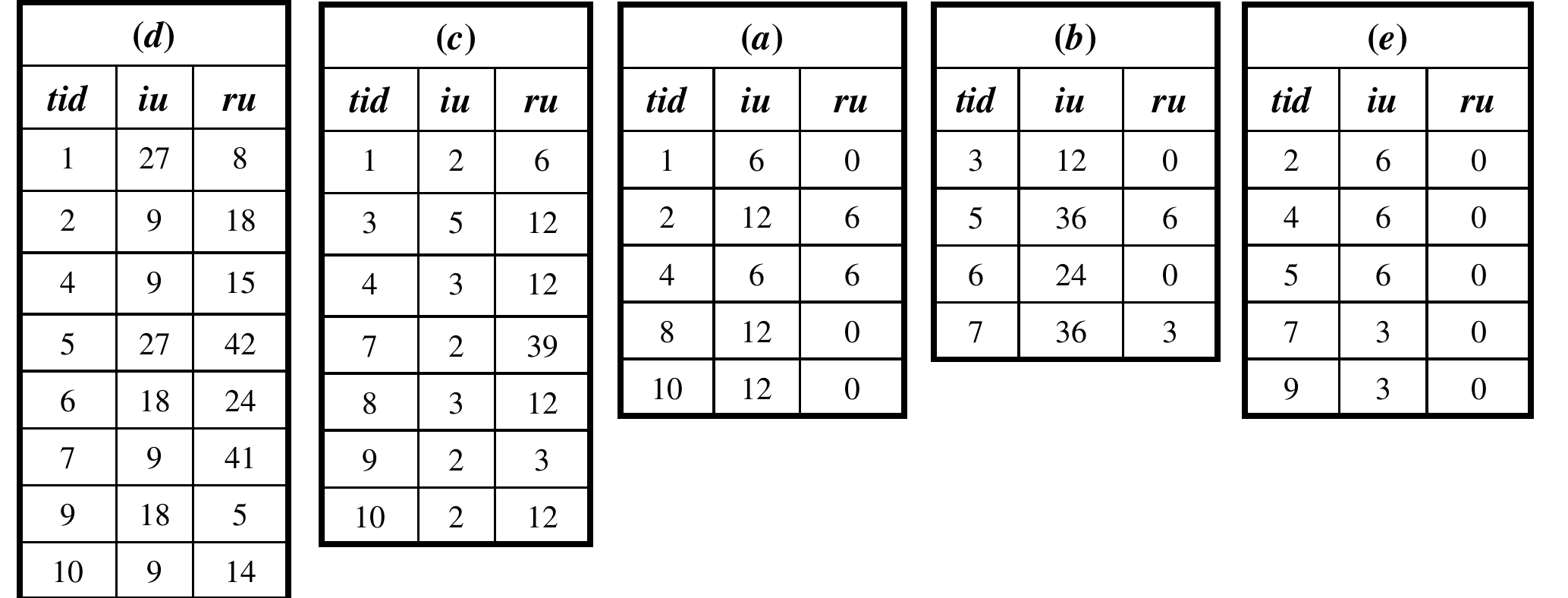}
	\captionsetup{justification=centering}
	\caption{Constructed utility-lists of 1-itemsets in the running example.}
	\label{fig:Fig1}
\end{figure*}

\begin{definition}
  \rm  Based on the definition of utility-list, let \textit{X.IU} be the sum of the utilities of an itemset $X$ in $D$, which can be defined as
  \begin{equation}
       X.IU=\sum_{X\subseteq T_{q}\wedge T_{q}\in D}X.iu(T_{q}).
  \end{equation}
\end{definition}
	
\begin{definition}	
  \rm Based on the designed utility-list structure, let \textit{X.RU} be the sum of the remaining utilities of an itemset $X$ in $D$, which can be defined as   
  \begin{equation}
        X.RU=\sum_{X\subseteq T_{q}\wedge T_{q}\in D}X.ru(T_{q}).
  \end{equation}
\end{definition}

The itemset $(b)$ is used as an example to illustrate the two definitions. The itemset $(b)$ exists in transactions $\{T_3, T_5, T_6, T_7\}$. \textit{b.IU} is calculated as (12 + 36 + 24 + 36) = \$108, and \textit{b.RU} is calculated as (0 + 6 + 0 + 3) = \$9. The itemset $(be)$ exists in transactions $\{T_5, T_7\}$; \textit{be.IU} = \textit{b.IU} + \textit{e.IU} = (36 + 36) + (6 + 3) = \$81, and \textit{be.RU} = \textit{b.RU} + \textit{e.RU} = (6 + 3) + (0 + 0) = \$9.

In the designed utility-list structure, we can effectively avoid the multiple database scans. For example, for the itemset $(b)$, it appears in transactions $T_3$, $T_5$, $T_6$, and $T_7$. The utility-list structure of $(b)$ is constructed and shown in Fig. \ref{fig:Fig3}(a). The accumulated utilities and remaining utilities are simultaneously calculated during the construction process of the utility-list structure of the item $(b)$. In this example, the $u(b)$ can be calculated as \textit{b.IU} = $u(b)$ = 12 + 36 + 24 + 36 = \$108, and the remaining utility of $(b)$ can be calculated as \textit{b.RU} = 0 + 6 + 0 + 3 = \$9. The obtained \textit{information-table} (\textit{I-table}) of the item $(b)$ is shown in Fig. \ref{fig:Fig3}(b).

\begin{figure*}[hbtp]	
	\centering
	\includegraphics[trim = 0 0 0 0,clip,scale=0.45]{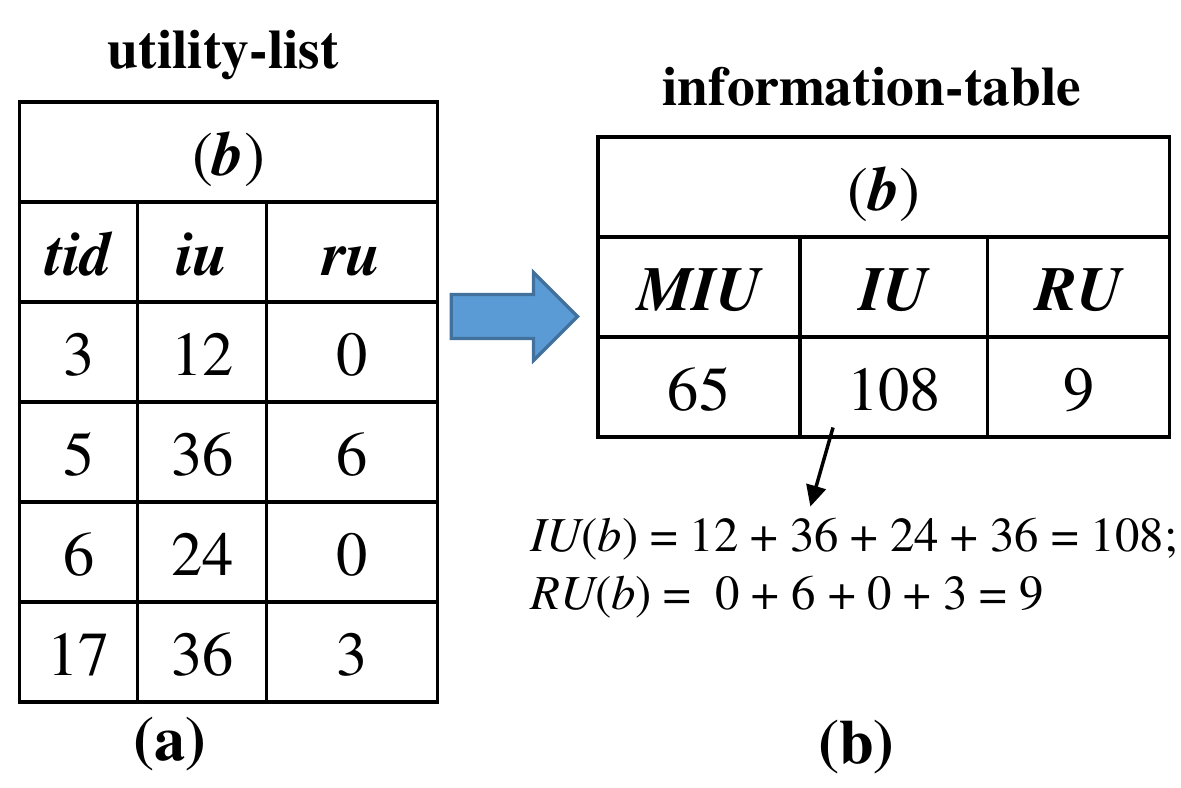}
	\captionsetup{justification=centering}
	\caption{Obtained information-table of the item $(b)$.}
	\label{fig:Fig3}
\end{figure*}

Similarly, the obtained information-tables of all 1-itemsets from the given example are shown in Fig. \ref{fig:Fig4}.

\begin{figure*}[hbtp]	
	\centering
	\includegraphics[trim = 0 0 5 0,clip,scale=0.35]{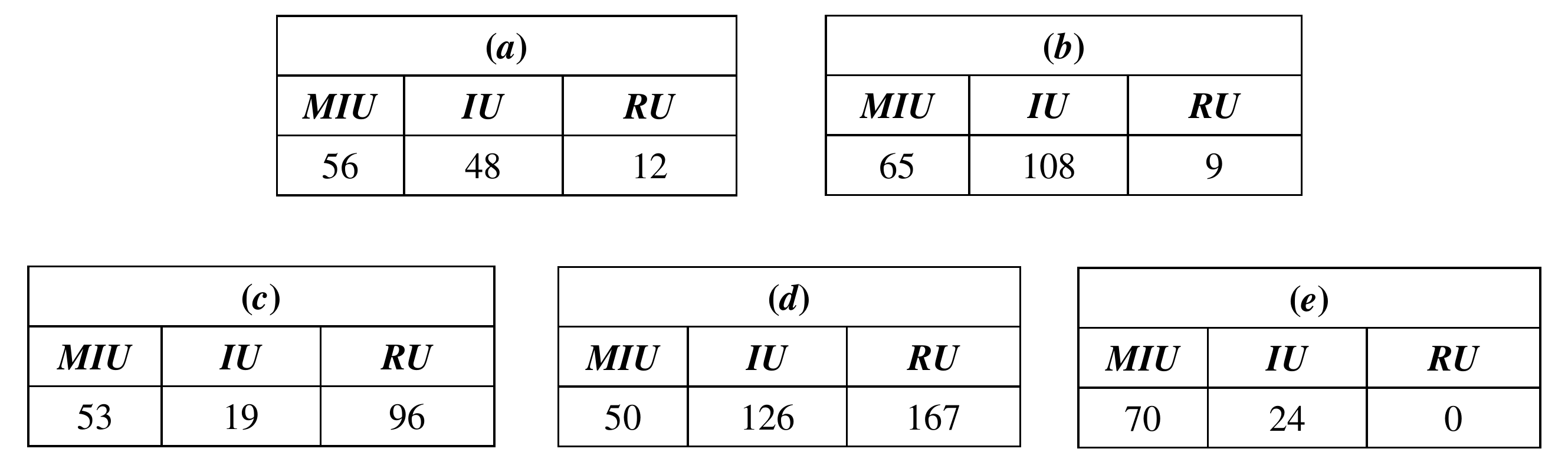}
	\captionsetup{justification=centering}
	\caption{Obtained information-tables of all 1-itemsets.}
	\label{fig:Fig4}
\end{figure*}

Thus, we do a depth-first search of a set enumeration tree. The traversing principle in the developed MIU-tree is stated below: Each node $X$ in the MIU-tree consists of its utility determination. If the utility of the currently processed node $X$ is no less than its \textit{MIU} value, and the summation of \textit{X.IU} and \textit{X.RU} is no less than the \textit{MIU} value, the extensions of the node $X$ are required to be generated and determined; otherwise, the depth-search produce will be stopped. This procedure can be used to efficiently reduce the search space of the MIU-tree. Based on the developed MIU-tree, the actual HUIs can be directly discovered without multiple database scans. The \textit{SDC} property of HUIs thus can be used as the conditional anti-monotone property to prune the unpromising child nodes with the \textit{MIU}, thus speeding up the combination process of utility-list structure and reducing the search space. Based on the previous definitions, two early pruning strategies are used to find a more compressed search space according to the TWDC property. These properties allow us to prune a great number of unpromising itemsets, thus significantly reducing the search space for mining the HUIs with varied item-specific minimum utility thresholds.

\begin{strategy}
	\label{strategy1}
	When traversing the MIU-tree using a depth-first search, if the TWU of a node $X$ based on its  utility-list is less than the \textit{LMU}, then none of the supersets of node $X$ (note that here supersets contains not only descendant nodes of $X$, but also other nodes having $X$ as subset) are HUIs.
\end{strategy}

According to Theorem 2, we can observe that if the TWU of a node is less than the \textit{LMU}, this node is not an HUI, as well as all its superset node. As shown in Fig. \ref{MIU_tree}, assume that the node ($e$) has $TWU(e) < LMU$, then all of the superset nodes of itemset ($e$) would not be regarded as an HUI because their TWU values are always no greater than the \textit{LMU}. Hence, all of the nodes containing ($e$) are directly skipped without constructing their utility-list structures.

\textbf{Theorem 4.} (\textit{Conditional downward closure property}, \textit{CDC property})
	emph{For any node $X$ in the MIU-tree, the sum of $X.IU$ and $X.RU$ in the utility-list of $X$ is no less than the utility of any one of its descendant nodes (extensions). Thus this sum is anti-monotonic and allows pruning  itemsets in the MIU-tree.}

\begin{proof}
    Let $ X^{k-1} $ be a (\textit{k}-1)-itemset, and $ X^{k} $ be a (\textit{k})-itemset that is an extension of $ X^{k-1}$. Assume that $ X^{k} $ is a children of $ X^{k-1} $ in the MIU-tree, meaning that $ X^{k-1} $ is a prefix of $ X^{k}$. Let the set of items in $ X^{k} $ but not in $ X^{k-1} $ be denoted as ($ X^{k} $$-$$ X^{k-1} $) = ($X^{k}$$\setminus $$X^{k-1}$), and the set of all the items appearing after $ X^{k} $ in transaction \textit{T} is denoted as $ T/X^{k} $. For any transaction $X^{k}\subseteq T_{q}$:
\begin{tabbing}
	$\because X^{k-1}\subset X^{k}\subseteq T_{q}\Rightarrow (X^{k}\setminus X^{k-1})\subseteq(T_{q}\setminus X^{k-1}).$\\
	$\therefore$ \text{In} each $ T_{q}, X^{k}.iu $ \=$ =X^{k-1}.iu+(X^{k}\setminus X^{k-1}).iu$
	$=X^{k-1}.iu+\sum_{z\in(X^{k}\setminus X^{k-1})}z.iu$\\
	
	$\therefore X^{k}.iu$ \= $\leq X^{k-1}.iu+\sum_{z\in(T_{q}/X^{k-1})} z.iu$
	$=X^{k-1}.iu+X^{k-1}.ru$\\
	
	$\therefore$ \text{In} each $ T_{q}, X^{k}.iu\leq X^{k-1}.iu+X^{k-1}.ru$\\
	$\because X^{k-1}\subset X^{k}\Rightarrow X^{k}.tids\subseteq X^{k-1}.tids$\\
	$\therefore \text{in }D, X^{k}.IU$\=$=\sum_{T_{q}\in X^{k}.tids}X^{k}.iu$
	$\leq \sum_{T_{q}\in X^{k}.tids}(X^{k-1}.iu+X^{k-1}.ru)$\\
	
	\>$\leq \sum_{T_{q}\in X^{k-1}.tids}(X^{k-1}.iu+X^{k-1}.ru)$
	$= X^{k-1}.IU+X^{k-1}.RU$\\
	
	$\therefore \text{in }D, X^{k}.IU$ \= $ \leq X^{k-1}.IU+X^{k-1}.RU$.
\end{tabbing}	
Thus, for the addressed HUIM-MMU problem, the sum of the utilities of $ X^{k} $ in \textit{D} is no greater than (\textit{$X^{k-1}.IU$} + \textit{$X^{k-1}.RU$}) of $ X^{k-1}$ in \textit{D}.

\end{proof}

\begin{strategy}
	\label{strategy2}
	When traversing the MIU-tree using a depth-first search, if the sum of \textit{X.IU} and \textit{X.RU} in the  utility-list of an itemset $X$ is less than \textit{MIU(X)}, then none of the descendant nodes (extensions) of node X is an HUI since the actual utilities of these extensions will be less than \textit{MIU(X)}.
\end{strategy}

In the running example, assume that the node $ (e) $ has $ TWU(e) < LMU$. Then the visited nodes, pruned nodes, and the skipped nodes are respectively  shown in Fig. \ref{MIU_tree}(right) when applying the \textbf{Strategy 1}. And the \textbf{Strategy 2} is used as a conditional strategy to prune all extensions of an unpromising node early.

\subsection{Proposed HIMU Algorithm}
 As shown in the main procedure of HIMU (Algorithm 1), the proposed HIMU algorithm initially sets \textit{i.UL} and \textit{D.UL} as an empty set (Line 1) and calculates the \textit{LMU} from MMU-table (Line 2). The database is scanned to calculate the $TWU(i)$ value of each item $ i \in I $ (Line 3), and the potential 1-itemsets were found such that $ TWU(i) \geq LMU ( I^*  \subseteq HTWUI^1)$ (Line 4). After sorting $I^*$ in the total order $\prec$ (ascending order of \textit{MIU} value), the HIMU algorithm scans $D$ again to construct the utility-list structure of each 1-itemset $ i \in I^* $, and then get the set of $D.UL$ for all 1-itemsets (Lines 5-6). It is important to note that only the designed order $\prec$ can guarantee the completeness of the output results, and the reason has been noted. The utility-list structures for all 1-extensions of $ i \in I^* $ are recursively processed by using a depth-first search procedure HUI-Search (Line 7). Details of the HUI-Search procedure are shown in Algorithm 2.

%%%%%%%%%%%algorithm 1%%%%%%%%%%%%%%%
\begin{algorithm}
	\LinesNumbered
	%	\scriptsize
	\caption{The $ \rm  HIMU $ algorithm}
	\KwIn{\textit{D}; \textit{ptable}; \textit{MMU-table} = $ \{mu(i_{1}), mu(i_{2}), \dots, mu(i_{m})\}$.}
	\KwOut{The set of complete high-utility itemsets (HUIs).}
	initialize $ i.UL \gets \emptyset, D.UL \gets \emptyset $\;
	calculate the $LMU$ in the \textit{MMU-table}\;
	scan $D$ to calculate the  $ TWU(i) $ value of each item $ i\in I$\;
	find $ I^* \gets \left\{  i \in I | TWU(i) \geq LMU \right\}$, w.r.t. $ HTWUI^1 $\;
	sort $ I^* $ according to the designed total order $ \prec $\ (ascending order of $mu$ values)\;
	scan $D$ to construct the utility-list \textit{i.UL} of each item  $ i\in I^{*} $, and put them into the set of $D.UL$\;
	\textbf{call HUI-Search}($\phi, I^*, \textit{MMU-table}$)\;
	\Return \textit{HUIs}\
\end{algorithm}
%%%%%%%%%%%%%%%%%%%%%%%%%%%%%%%%%%%%%

%%%%%%%%%%%algorithm 2%%%%%%%%%%%%%%%
\begin{algorithm}
	\LinesNumbered
	%	\scriptsize
	\caption{The HUI-Search Procedure}
	\KwIn{$X$, $\textit{extensionsOfX}$, \textit{MMU-table}.}
	\KwOut{The complete set of HUIs.}
	\For {each itemset $ X_{a}\in $ $ \textit{extensionsOfX} $}
	{
		obtain the $ X_{a}.IU $ and $ X_{a}.RU $ values from the built $ X_{a}.UL $\;
		\If{$ X_{a}.IU\geq MIU(X_a) $}
		{
			$ HUIs\leftarrow HUIs\cup X_{a} $\;	
		}
		\If {$ (X_{a}.IU + X_{a}.RU\geq MIU(X_a)) $}
		{
			$ \textit{extensionsOfX}_{a}\leftarrow  \emptyset $\;
			\For {each itemset $ X_{b}\in \textit{extensionsOfX} $ such that $ X_{b} $  after  $ X_{a} $}
			{			
				$ X_{ab}\leftarrow X_{a} \cup X_{b} $\;
				$ X_{ab}.UL\leftarrow construct(X, X_{a}, X_{b}) $\;
				\If{$ X_{ab}.UL \not= \emptyset$}
				{
					$ \textit{extensionsOfX}_{a}\leftarrow \textit{extensionsOfX}_{a}\cup X_{ab}.UL $\;	
				} 				
			}   		
			\textbf{call HUI-Search}$\boldmath{(X_{a}, \textit{extensionsOfX}_{a}, MIU(X_a), EUCS)}$\;	
		} 
	}
	\textbf{return} \textit{HUIs}\
\end{algorithm}

As shown in HUI-Search procedure (cf. Algorithm 2), the necessary information w.r.t. $X_{a}.IU$, $X_{a}.RU$, and the $MIU(X_a)$ values from the built structures (Line 2) are obtained. After that, each itemset $X_a$ is determined whether it is the desired HUI (Lines 3-5). We then apply the pruning strategy 2 to further determine whether its child nodes should be processed for the later depth-first search (Lines 6-16). If any extensions of one node is calculated as $X_{a}.IU$ + $X_{a}.RU \geq MIU(X_a)$ (Line 6), the construction process $ Construct(X, X_a, X_b) $ is then continuously executed to construct a set of utility-list structures of all 1-extensions of itemset $X_a$  (w.r.t. $\textit{extensionsOfX}_a $) (Lines 8-14). Note that each constructed $X_{ab}$  is a 1-extension of itemset $X_a$ , if the utility-list structure of it is not empty (Line 11); it should be put into the set of $\textit{extensionsOfX}_a $ for executing the later depth-first search (Line 12). The designed HUI-Search procedure is recursively performed to mine HUIs (Line 15), and it finally returns the set of HUIs (Line 18). Based on the designed MIU-tree and the spanning mechanism, the HIMU algorithm can thus directly mine the complete set of HUIs from the database without candidate generation.

\section{Proposed Enhanced HIMU$_{EUCP}$, HIMU$_{LAP}$, and HIMU$_{ELP}$ Algorithms}
\label{sec:HIMU2}

In this section, three improved algorithms, namely, the enhanced HIMU$_{EUCP}$, HIMU$_{LAP}$, and HIMU$_{ELP}$ approaches with early pruning strategy, are further proposed to improve the performance of the designed HIMU algorithm. As shown in the previous studies, such as FHM \cite{fournier2014fhm} and HUP-Miner \cite{krishnamoorthy2015pruning}, two effective pruning strategies called estimated utility co-occurrence pruning (EUCP) and LA-Prune (LAP) can effectively early prune the redundant candidates for speeding up the computations. Details of the EUCP and LAP strategies are described below.

\subsection{EUCP and LAP Strategies}

The EUCP strategy is designed to avoid the construction of utility-list structures for unpromising patterns since construction of the utility-list structure (a join) is an expensive operation. It has been shown that the FHM algorithm \cite{fournier2014fhm} outperforms the HUI-Miner algorithm in most cases. The EUCP can be used to directly eliminate a low-utility pattern and all its transitive extensions without constructing their utility-list. Based on the HUIM-MMU framework and Definition 9, the following property can be obtained.

\textit{Property.} \textit{If a 2-itemset $X$ is not an HTWUI w.r.t. $TWU(X) \geq MIU(X)$, any $k$-itemset ($k \geq 3$) of $X$ will not be an HTWUI or HUI either}.

Note that the set HTWUI is determined using the $TWU(X) \geq MIU(X)$ criterion, which is different than the criterion used in FHM \cite{fournier2014fhm}. Based on the modified EUCP strategy, a huge number of unpromising $k$-itemsets ($k \geq 3$) can be efficiently pruned. To apply the EUCP strategy, a structure named estimated utility co-occurrence structure (EUCS) is first built. It is a matrix that stores the TWU values of all 2-itemsets. The constructed EUCS of the running example is shown in Table \ref{table_EUCS}.

\begin{table}[!htbp]
	\centering
	\small
	\caption{EUCS of given example.} 
	\label{table_EUCS}
	\begin{tabular}{|c|c|c|c|c|}
		\hline
		\textbf{Item} & \textbf{$a$} & \textbf{$b$}  & \textbf{$c$}  & \textbf{$d$} \\ \hline
		\textbf{$b$}   &   \$0   &   -  &  -  &  -	    \\ \hline
		\textbf{$c$}   &   \$97   &   \$91  &  -  & -	    \\ \hline
		\textbf{$d$}   &   \$109   &  \$135  &  \$155   &  - \\ \hline
		\textbf{$e$}   &   \$51   &   \$95   &  \$97    &  \$169 \\ \hline
	\end{tabular}
\end{table}

Based on the join operation of utility-list structure in HIMU algorithm, a superset of a non-HUI may still be an HUI; only the \textit{CDC} property is guaranteed in the designed MIU-tree. Thus, we try to utilize the TWU value instead of the \textit{minutil} value as the pruning condition when using the EUCP strategy. Thus, the \textit{global downward closure} property can be held. Theorem 2 shows that all supersets (not only the child nodes of a prefix node, but also its supersets that are not its extension's descendant nodes) cannot be HUIs if any one of their subsets is not an HTWUI according to the \textit{LMU}.

\begin{strategy}[EUCP Strategy \cite{fournier2014fhm}] % (estimated utility co-occurrence pruning, EUCP strategy )
	If the TWU of a 2-itemset $X$ in EUCS is less than the \textit{MIU} value, any superset of $X$ will not be an HUI; it can be directly pruned.
\end{strategy}

Hence, the EUCP strategy supports a powerful capability for early pruning those unpromising $k$-itemsets ($k \geq 3)$. To perform this strategy, each 2-itemset in the EUCS is then checked to determine whether its TWU is less than the \textit{MIU} of its prefix node.

Recently, a new algorithm, namely HUP-Miner, was proposed as an improvement to HUI-Miner \cite{krishnamoorthy2015pruning}. Two new pruning strategies, PU-Prune and LA-Prune, are introduced in HUP-Miner \cite{krishnamoorthy2015pruning} to reduce the search space for mining HUIs. The LA-Prune property can provide a tighter utility upper bound for any itemset $Px$ containing $Py$. It is applied in this paper when the construction process is performed on $Pxy$ from $Px$ and $Py$. We further extend the LA-Prune strategy to the proposed HIMU algorithm.

\begin{lemma}
  Given two itemsets $X$ and $Y$ ($X  Y$), neither $\{X, Y\}$ nor any supersets of $X$ will be an HUI if \textit{X.IU} + \textit{X.RU} -$ \sum _{X\subseteq T_{q} \wedge T_{q}\subseteq D \wedge Y \nsubseteq T_{q}}((X.iu + X.ru) \leq MIU(X)) $ .  % \neq
\end{lemma}
\begin{strategy}[LA-Prune Strategy \cite{krishnamoorthy2015pruning}] 
   Let $X$ be a processed itemset (node) during the depth-first search of the set-enumeration MIU-tree, and $Y$ be the right sibling node of $X$. If the sum of \textit{X.IU} + \textit{X.RU} subtracts the utilities \textit{X.iu} + \textit{X.ru} of a set of transactions is less than \textit{MIU(X)}, $\{X,Y\}$ is not an HUI, and any of its child nodes are not an HUI. The construction of the utility-list structures for the children nodes of $X$ is not necessary to perform.
\end{strategy}

Based on the two proposed pruning strategies, the designed HIMU algorithm can prune the itemsets with lower utility early without constructing the utility-list structures of their extensions. Thus, it can effectively reduce both the computations of join operations and the search space from the MIU-tree.

\subsection{Proposed HIMU$_{EUCP}$, HIMU$_{LAP}$, and HIMU$_{ELP}$ Algorithms}

Based on the EUCP and LAP strategies, the proposed HIMU$_{EUCP}$, HIMU$_{LAP}$, and HIMU$_{ELP}$ algorithms, collectively, and the \textbf{HUI-Search-EUCP} procedure are described in Algorithm 3 and Algorithm 4, respectively. The improved HIMU$_{EUCP}$ algorithm is similar to the first designed HIMU algorithm, except that the EUCS is constructed in the second database scan and the mining procedure utilizes the EUCS for deriving HUIs. Details are described below.

%%%%%%%%%%%algorithm 1%%%%%%%%%%%%%%%
\begin{algorithm}
	\LinesNumbered
	%	\scriptsize
	\caption{The $ \rm  HIMU_{EUCP} $ algorithm}
	\KwIn{\textit{D}; \textit{ptable}; \textit{MMU-table} = $ \{mu(i_{1}), mu(i_{2}), \dots, mu(i_{m})\}$.}
	\KwOut{The set of complete high-utility itemsets (HUIs).}
	initialize $ i.UL \gets \emptyset, D.UL \gets \emptyset, EUCS \gets \emptyset $\;
	calculate the $LMU$ in the \textit{MMU-table}\;
	scan $D$ to calculate the  $ TWU(i) $ value of each item $ i\in I$\;
	find $ I^* \gets \left\{  i \in I | TWU(i) \geq LMU \right\}$, w.r.t. $ HTWUI^1 $\;
	sort $ I^* $ according to the designed total order $ \prec $\ (ascending order of $mu$ values)\;
	scan $D$ to construct the utility-list \textit{i.UL} of each item  $ i\in I^{*} $ and build the \textit{EUCS}\;
	\textbf{call HUI-Search}($\phi, I^*, \textit{MMU-table}, EUCS$)\;
	\Return \textit{HUIs}\
\end{algorithm}
%%%%%%%%%%%%%%%%%%%%%%%%%%%%%%%%%%%%%

%%%%%%%%%%%algorithm 2%%%%%%%%%%%%%%%
\begin{algorithm}
	\LinesNumbered
	%	\scriptsize
	\caption{The HUI-Search procedure with EUCP strategy}
	\KwIn{$X$, $\textit{extensionsOfX}$, \textit{MMU-table}, $EUCS$.}
	\KwOut{The complete set of HUIs.}
	\For {each itemset $ X_{a}\in $ $ \textit{extensionsOfX} $}
	{
		obtain the $ X_{a}.IU $ and $ X_{a}.RU $ values from the built $ X_{a}.UL $\;
		\If{$ X_{a}.IU\geq MIU(X_a) $}
		{
			$ HUIs\leftarrow HUIs\cup X_{a} $\;	
		}
		\If {$ (X_{a}.IU + X_{a}.RU\geq MIU(X_a)) $}
		{
			$ \textit{extensionsOfX}_{a}\leftarrow  \emptyset $\;
			\For {each itemset $ X_{b}\in \textit{extensionsOfX} $ such that $ X_{b} $  after  $ X_{a} $}
			{
				\If{$ \exists   TWU(a,b)\in  EUCS  \wedge  TWU(a,b)\geq MIU(X_a) $}
				{
					$ X_{ab}\leftarrow X_{a} \cup X_{b} $\;
					$ X_{ab}.UL\leftarrow $ \textbf{construct}($X, X_{a}, X_{b}$) \;
					\If{$ X_{ab}.UL \not= \emptyset$}
					{
						$ \textit{extensionsOfX}_{a}\leftarrow \textit{extensionsOfX}_{a}\cup X_{ab}.UL $\;	
					} 			
				}	
			}   		
			\textbf{call HUI-Search}$\boldmath{(X_{a}, \textit{extensionsOfX}_{a}, MIU(X_a), EUCS)}$\;	
		} 
	}
	\textbf{return} \textit{HUIs}\
\end{algorithm}

As shown in the \textbf{HUI-Search-EUCP} procedure (cf. Algorithm 4), the EUCP strategy can avoid performing a huge number of join operations of utility-list structures for those unpromising $k$-itemsets ($k \geq 3$) (Lines 7-11). Based on the constructed EUCS and the EUCP strategy, the improved HIMU$_{EUCP}$ algorithm can speed up the computations for mining HUIs.

The improved HIMU$_{LAP}$ algorithm is similar to the first designed HIMU algorithm, except for the utility-list construction procedure. It utilizes the LAP strategy to avoid constructing a huge number of utility-list structures for those unpromising itemsets. Details of the construction procedure are described below.

\begin{algorithm}[h]
	\LinesNumbered
	\SetKwInOut{Input}{input}\SetKwInOut{Output}{output}
	\Input{$P$: an itemset, $Px$: the extension of $P$ with an item $x$, $Py$: the extension of $P$ with an item $y$ ($X \neq y$), \textit{MIU(Px)}: the \textit{MIU} value of $Px$}
	\Output{\textit{Pxy.UL}: the utility-list of $Pxy$}
	\BlankLine
	$Pxy.UL \leftarrow \emptyset $\;
	Set \textit{Utility} = \textit{P.IU} + \textit{P.RU}\;
	\ForEach{ element/tuple $ex \in Px.UL$}{
		\If{$\exists ey \in Py.UL$ and $ex.tid = exy.tid$}{
			\If{$P.UL \neq \emptyset$ }{
				Search element $e \in P.UL$ such that $e.tid = ex.tid$.\;
				$exy \leftarrow <ex.tid, ex.iu+ey.iu - e.iu, ey.ru>$\;
			}
			\Else{
				$exy \leftarrow <ex.tid, ex.iu+ey.iu, ey.ru>$\;
			}
			$Pxy.UL \leftarrow Pxy.UL \cup \{exy\}$\;
		}
	    \Else{
	    	\textit{Utility} = \textit{Utility} - \textit{ex.iu} - \textit{ex.ru}\;
	    	\If{$Utility < MIU(Px)$ }{
	    		\textbf{return} \textit{NULL}\;
	    	}
     	}
	
	}
	\Return{Pxy.UL}\;
	\caption{utility-list construction with LA-Prune}
	\label{AlgorithmConstruct}
\end{algorithm}

Based on the developed HIMU$_{EUCP}$ and HIMU$_{LAP}$ algorithms, the improved HIMU$_{ELP}$ algorithm is a hybrid algorithm that adopts both the EUCP strategy (which is applied in HUI-Search procedure) and LAP strategy (which is applied in the utility-list construction procedure). Therefore, details of the hybrid HIMU$_{ELP}$ algorithm are the same as shown in Algorithm 3, Algorithm 4 and Algorithm 5.

\begin{lemma}
   The HIMU algorithm and its enhanced versions are correct and complete for discovering the set of HUIs while considering multiple minimum utility thresholds.
\end{lemma}

\begin{proof}
   The HIMU algorithm first examines the transactions in the transactional data to identify the promising items, and, at the same time, calculates the \textit{tid}, \textit{iu}, and \textit{ru} values of each item. Next, the utility-lists of promising items are constructed. When traversing the set-enumeration MIU-tree to mine HUIs, the utility of each node is accurately evaluated using its utility-list. By Lemmas 1-5, the HIMU algorithm ensures that only unpromising itemsets are ignored, which guarantees the completeness of the discovered HUIs, and that the utility of each itemset can be accurately calculated using the utility-list structure to ensure the correctness of the obtained HUIs. Therefore, the proposed HIMU algorithm holds the correctness and completeness based on the utility-list and MIU-tree structure.
\end{proof}

\section{Experimental Results} \label{sec:6}

In this section, we evaluated the performance of the proposed one-phase HIMU algorithm and its enhanced algorithms for mining high-utility itemsets with multiple minimum utility thresholds in four datasets, including different real-world datasets and one synthetic dataset. We implemented the well-know traditional HUIM algorithm FHM \cite{fournier2014fhm} (for mining HUIs with an uniform minimum utility threshold) and the state-of-the-art HUI-MMU and HUI-MMU$_{TID}$ algorithms \cite{lin2016efficient} (for mining HUIs with multiple minimum utility thresholds) as the benchmarks against which the designed efficient algorithms will be measured. It is important to notice that all the traditional HUIM algorithms (e.g., HUI-Miner, FHM, d2HUP, etc.) which discover HUIs with an uniform minimum utility threshold are not suitable to be compared to evaluate the efficiency (w.r.t. execution time, memory consumption) since the addressed mining tasks are different. Thus, we use the existing HUI-MMU and HUI-MMU$_{TID}$ algorithms to evaluate the mining efficiency of the proposed several algorithms. Besides, we use the FHM algorithm to discover the high-utility itemsets using  an uniform minimum utility threshold, and then to evaluate the effectiveness of the proposed models with results of pattern analysis.

\subsection{Experimental Setup and Dataset Description}
% HIMU$_{EUCP}$, HIMU$_{LAP}$, and HIMU$_{ELP}$

All algorithms in the experiments are implemented in the Java language and executed on a PC with an Intel Core i5-3470 processor at 3.2 GHz CPU and 4 GB of RAM, running on the 64-bit Windows 7 platform. Note that the performance of the different algorithms may depend on the varied characteristics of the datasets. Therefore, we used not only sparse and dense datasets, but also real-world and synthetic datasets, to evaluate the effectiveness of the proposed algorithms. Experiments are conducted on four datasets, including three real-world datasets (foodmart \cite{foodmart}, retail \cite{fimdatasets}, and mushroom \cite{fimdatasets}) and one simulated dataset, T10I4D100K, which was generated by IBM Quest Synthetic Data Generation \cite{IBMdata} in terms of running time, number of HUIs, and scalability. The foodmart dataset contains real utility values, while a simulation model \cite{liu2005two,tseng2010up,lin2016efficient} was developed to generate the quantities and profit values of items in transactions for the retail, mushroom, and T10I4D100K datasets, by choosing random values in the [1,5] and [1,1000] intervals, respectively. The characteristics of these datasets are shown in Table \ref{characteristics}.

%%%%%%%%%%%%%%%%%%%%%%%% TABLE 1 - DATASET CHARACTERISTICS
\begin{table}[ht!]
	\caption {Dataset characteristics}
	\label{characteristics}
	\begin{center}
		\begin{tabular}{|c|c|c|c|c|c|}
			\hline
			\textbf{Dataset} &  \textbf{\# of transactions}  &  \textbf{\# of distinct item}s  &  \textbf{Avg. length}   & \textbf{Max. length}  \\
			\hline \hline
            \textit{foodmart}  &  21,556 &   1,559   & 4   & 11 \\
			\textit{retail}  &  88,162 &   16,470   & 10.3   & 76 \\
			\textit{mushroom}  &  8,124 &   119   & 23   & 23 \\
			\textit{T10I4D100K} &	100,000 &	870 &	10.1  &  29  \\
			\hline
		\end{tabular}
	\end{center} 
	\label{table-datasets}
\end{table}
%%%%%%%%%%%%%%%%%%%%%%%%  END OF TABLE 1

The foodmart dataset is from an anonymous chain store provided by Microsoft SQL Server. The retail dataset includes approximately 5 months of receipts from a supermarket. The mushroom dataset has 8,124 transactions with 119 distinct items, and the lengths of average transaction and maximal transaction respectively are 23 and 23, respectively. The T10I4D100K dataset was generated by the well-known IBM Quest Synthetic Data Generation \cite{IBMdata}. The source code of the FHM algorithm and the tested datasets can be downloaded from the SPMF data mining library (http://www.philippe-fournier-viger.com/spmf/) \cite{fournier2016spmf}.

Furthermore, based on the method discussed in \cite{liu1999mining} to assign multiple thresholds to items in FPM, we defined a method to automatically set the \textit{mu} value of each item in our proposed framework as:
\begin{equation}
    mu(i_{j}) = max[\beta \times pr(i_{j}), GLMU],
\end{equation}

where $\beta$ is a constant used to set the mu values of items as a function of their profit values \cite{lin2016efficient}. To ensure randomness and equipment diversity, $\beta$ was set in the [1,100] interval for the foodmart dataset, and set in the [1000,10000] interval for the other three datasets. The parameter \textit{GLMU} is the user-specified global least minimum utility value, and $ pr(i_{j})$ is the external utility of an item $ pr(i_{j})$. Note that if $\beta$ is set to zero, then a single minimum utility value \textit{GLMU} will be used for all items, and this will be equivalent to traditional HUIM.

\subsection{Runtime Analysis}

We first compare the runtimes of the all algorithms to evaluate their efficiency. For the conducted experiments under various \textit{GLMU}, the parameter $\beta$ was randomly set to be a fixed number of each item. The algorithm is terminated if the execution time exceeds 10,000 s. The runtime of the compared algorithms under various \textit{GLMU} with a fixed $\beta$ on the four datasets is shown in Fig. \ref{fig_Runtime1}.

\begin{figure*}[hbtp]	
	\centering
	\includegraphics[trim=10 0 5 0,clip,scale=0.38]{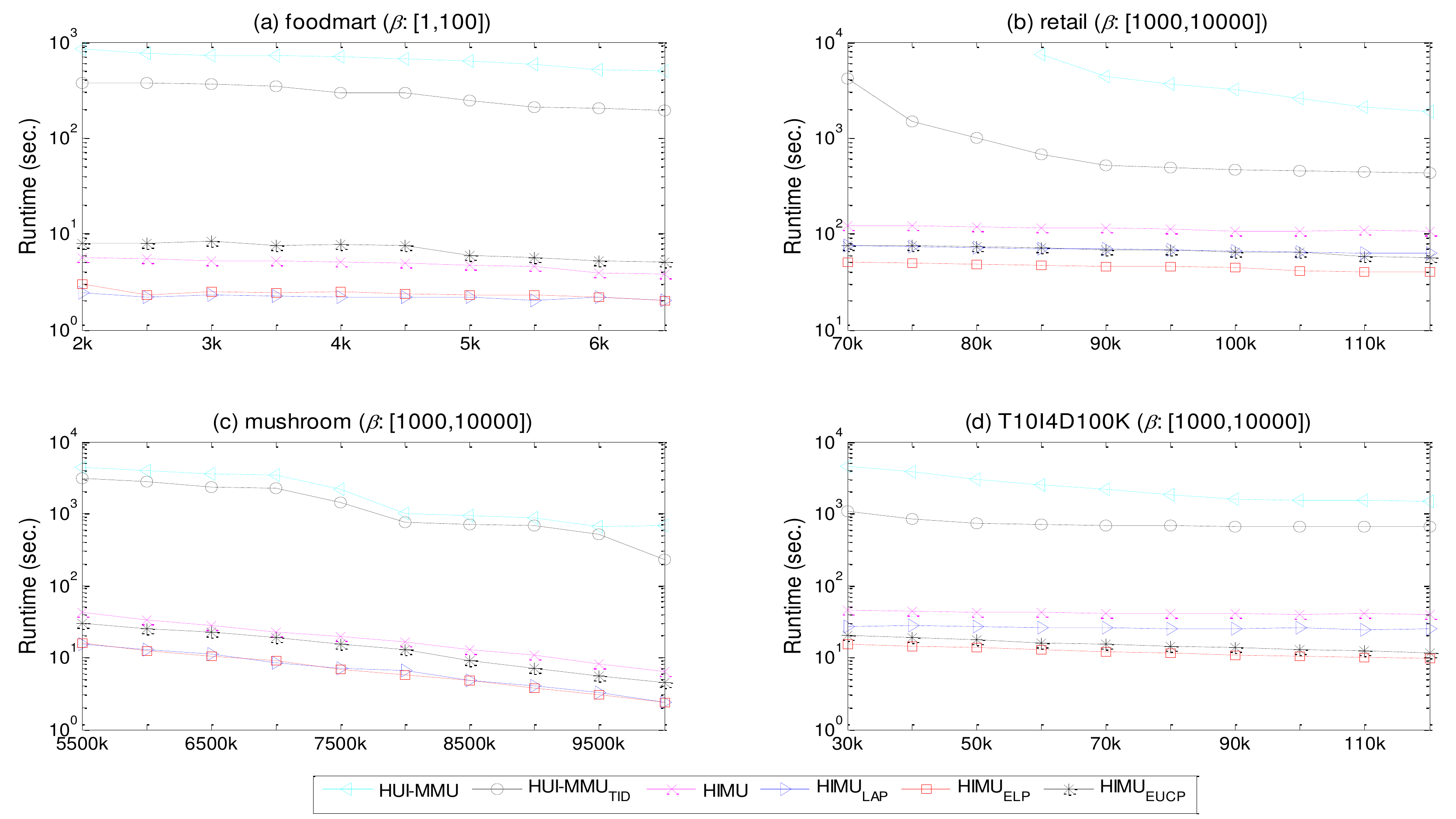}
	\captionsetup{justification=centering}
	\caption{Runtimes under a fixed $\beta$ and various \textit{GLMUs}.}
	\label{fig_Runtime1}
\end{figure*}

From Fig. \ref{fig_Runtime1}, it can be seen that the HIMU and three improved algorithms, including HIMU$_{EUCP}$, HIMU$_{LAP}$, and HIMU$_{ELP}$, performed well compared with the HUI-MMU and HUI-MMU$_{TID}$ algorithms. Moreover, all the proposed MIU-tree-based algorithms are generally up to almost 1 or 2 orders of magnitude faster than the level-wise HUI-MMU and HUI-MMU$_{TID}$ algorithms, and HIMU$_{ELP}$ is always faster than the other HIMU algorithms by adopting all early pruning strategies. This indicates that the generate-and-test approach has worse performance than the proposed set-enumeration tree-based approaches. For example, in Fig. \ref{fig_Runtime1}(c), the runtimes of the HUI-MMU and HUI-MMU$_{TID}$ algorithms dramatically decrease from 822 to 562 s, while those of the four proposed algorithms steadily decrease from 80 to 66 s. Because the downward closure property of Apriori is not held in HUIM algorithms, the previous HUI-MMU and HUI-MMU$_{TID}$ algorithms are necessary to overestimate the utility of candidates as the TWU upper-bound for maintaining the TWDC property. This process requires more computations, especially when numerous HTWUIs are generated, but fewer HUIs are discovered. In contrast, the proposed MIU-tree-based HIMU algorithm first finds the 1-itemsets that are no less than the \textit{GLMU}, and then, second, utilizes the developed GDC and CDC properties to efficiently prune the unpromising itemsets and to discover the HUIs without generation-and-test candidates in the level-wise manner. Moreover, the enhanced HIMU$_{EUCP}$, HIMU$_{LAP}$, and HIMU$_{ELP}$ algorithms adopt the EUCP and LAP strategies with early termination operations; thus, they can avoid the costly join operation for constructing the utility-list structures of those unpromising itemsets. The results for runtimes under various $\beta$ with a fixed \textit{GLMU} for different datasets are shown in Fig. \ref{fig_Runtime2}.

\begin{figure*}[hbtp]	
	\centering
	\includegraphics[trim=10 0 5 0,clip,scale=0.38]{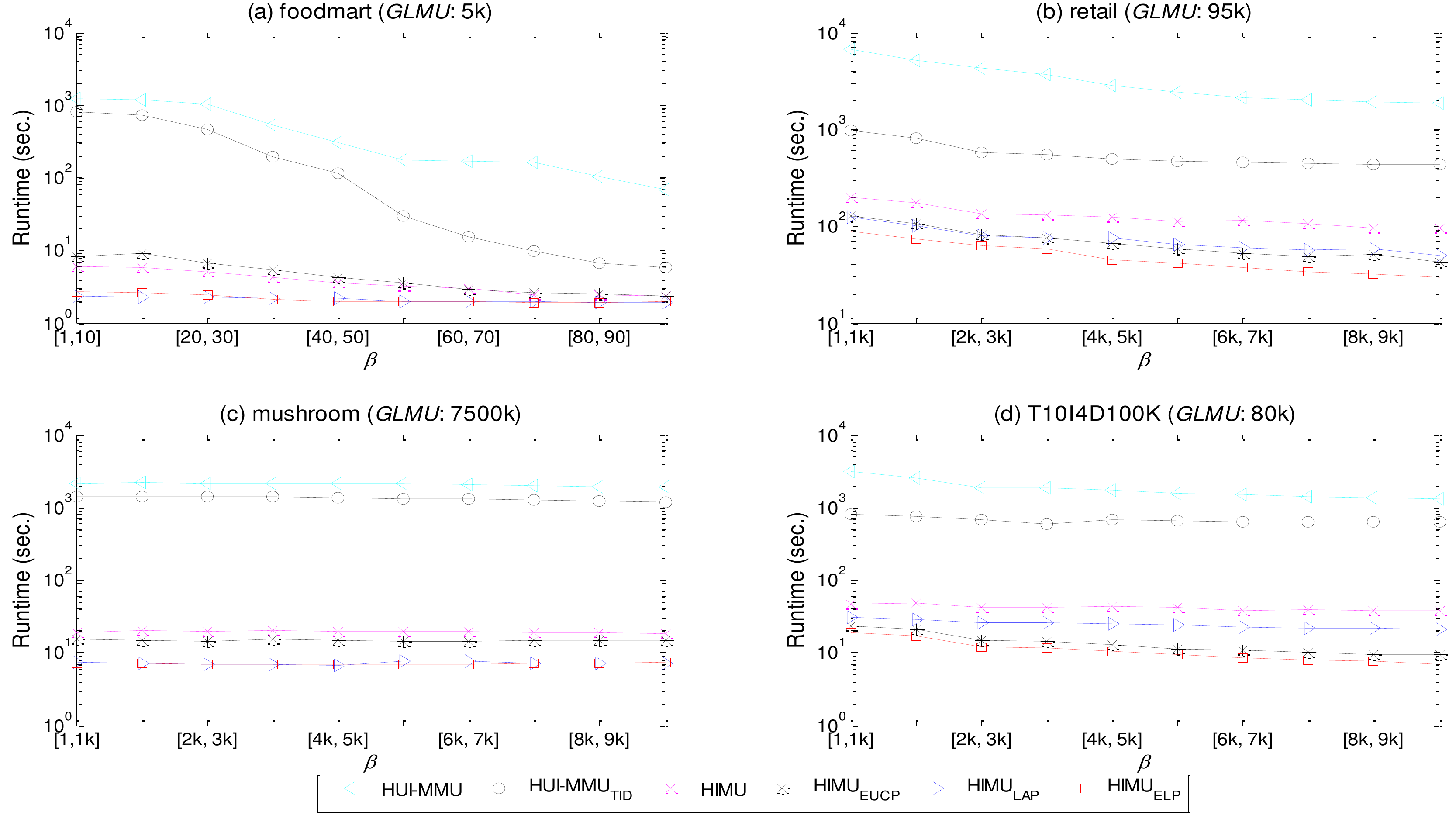}
	\captionsetup{justification=centering}
	\caption{Runtimes under a fixed \textit{GLMU} and various $\beta$.}
	\label{fig_Runtime2}
\end{figure*}

From Fig. \ref{fig_Runtime2}, it can be observed that the improved HIMU$_{EUCP}$, HIMU$_{LAP}$, and HIMU$_{ELP}$ algorithms also outperform the HUI-MMU, HUI-MMU$_{TID}$, and baseline HUI-MMU algorithms for the four datasets, and that HIMU$_{ELP}$ has the best performance among them. The reasons are the same as described for Fig. \ref{fig_Runtime1}. In particular, all the algorithms take less runtime to find the HUIs when $\beta$ is increased. The reason is that when $\beta$ is set as a larger value, the actual minimum utility (\textit{minutil}) threshold of each item is also set as a larger value based on the presented equation. Therefore, fewer HUIs and HTWUIs satisfy the defined condition in the level-wise HUI-MMU and HUI-MMU$_{TID}$ algorithms, and less execution time is required for them. Moreover, fewer patterns are generated by the MIU-tree-based algorithms with a larger \textit{minutil}, and the search space of the proposed algorithms can become more compressed than before. It is thus reasonable that both the previous level-wise approaches and the proposed algorithms consume less runtime when $\beta$ is increased.

In addition, the HIMU$_{EUCP}$ algorithm has the slightly longer runtime than the baseline HIMU algorithm, which can be observed in Fig. \ref{fig_Runtime2}(a), and HIMU$_{ELP}$ performs slightly worse than HIMU$_{LAP}$. The reason is that for the very sparse dataset, such as the foodmart dataset, many unpromising candidates can be directly pruned by the TWDC, GDC, and SDC properties. Thus, it is  not necessary to perform the EUCP strategy to construct the EUCS for pruning the unpromising itemsets, and the LAP strategy always has a positive impact on filtering unpromising itemsets, as shown in Fig. \ref{fig_Runtime1} and Fig. \ref{fig_Runtime2} under various parameters. Therefore, the HIMU$_{ELP}$ has the best performance among all compared algorithms, and it also indicates that the EUCP strategy has no efficient results on very sparse datasets.

\subsection{Pattern Analysis}

To evaluate the effectiveness and acceptability of the designed algorithms, this section evaluates the number of HUIs that are generated under a uniform minimum utility threshold and multiple minimum utility thresholds. Note that HUIs* are generated by the traditional HUIM algorithm (such as the FHM algorithm), and HUIs are generated by the HUIM-MMU framework (such as the designed algorithms), and HTWUIs are generated by the HUI-MMU algorithm. The number of generated patterns under various \textit{GLMUs} with a fixed $\beta$ is shown in Table \ref{table:patterns1}. We also evaluated how $\beta$ affects the number of derived HTWUIs and HUIs, resulting in Table \ref{table:patterns2}.

%%%%%%%%%%%%%%%%%%%%%%%%%%%%%%%%%%%%%%%%%%%%%%%%%%Friedman
\begin{table}[htb]
	\fontsize{6.5pt}{10pt}\selectfont
	\centering
	\caption{Derived patterns under varied \textit{GLMU}}
	\label{table:patterns1}
	\begin{tabular}{c|c|llllllllll}
		%\toprule
		\hline\hline
		\multirow{2}*{\textbf{Dataset}}&
		\multirow{2}*{\textbf{Pattern}}
		&\multicolumn{10}{c}{\textbf{Threshold}}\\
		\cline{3-12}%\morecmidrules
		& & \textit{GLMU}$_1$ &  \textit{GLMU}$_2$  &  \textit{GLMU}$_3$   &  \textit{GLMU}$_4$  &  \textit{GLMU}$_5$   &  \textit{GLMU}$_6$  &  \textit{GLMU}$_7$  &  \textit{GLMU}$_8$  &  \textit{GLMU}$_9$  &  \textit{GLMU}$_{10}$  \\ \hline
		
\textbf{foodmart} & \textit{GLMU} & 2000 & 2500 & 3000 & 3500 & 4000 & 4500 & 5000 & 5500 & 6000 & 6500 \\
($\beta$: [1, 100])	& HUIs* & 499,063 & 454,782 & 401,646 & 344,261 & 287,063 & 232,918 & 184,676 & 143,260 & 108,932 & 81,466 \\
		& HTWUIs & 335,201 & 334,413 & 333,015 & 330,147 & 325,732 & 319,032 & 308,860 & 296,138 & 277,644 & 255,590  \\
		& HUIs & 203,871 & 195,607 & 183,039 & 167,309 & 148,266 & 127,682 & 106,799 & 86,915 & 68,948 & 53,756  \\	    
		\hline
		
\textbf{retail} & \textit{GLMU} & 70000 & 75000 & 80000 & 85000 & 90000 & 95000 & 100000 & 105000 & 110000 & 115000 \\
($\beta$: [1000, 10000])	& HUIs* & 10,491 & 9,468 & 8,569 & 7,813 & 7,130 & 6,546 & 6,019 & 5,550 & 5,138 & 4,763 \\
& HTWUIs & 1,348,632 & 893,043 & 696,275 & 475,596 & 298,255 & 240,907 & 210,460 & 170,339 & 135,316 & 120,941  \\
& HUIs & 2,196 & 2,037 & 1,883 & 1,781 & 1,676 & 1,581 & 1,503 & 1,437 & 1,372 & 1,324  \\	    
\hline
		
\textbf{mushroom} & \textit{GLMU} &  5500000 & 6000000 & 6500000 & 7000000 & 7500000 & 8000000 & 8500000 & 9000000 & 9500000 & 10000000  \\
($\beta$: [1000, 10000]) & HUIs* & 20,524 & 14,517 & 9,459 & 5,615 & 2,975 & 1,406 & 579 & 195 & 44 & 3 \\
& HTWUIs & 652,193 & 621,539 & 606,135 & 569,003 & 384,117 & 139,483 & 121,753 & 111,789 & 75,391 & 70,607 \\
& HUIs & 20,524 & 14,517 & 9,459 & 5,615 & 2,975 & 1,406 & 579 & 195 & 44 & 3   \\
\hline
		
\textbf{T10I4D100K} & \textit{GLMU} &  30000 & 40000 & 50000 & 60000 & 70000 & 80000 & 90000 & 100000 & 110000 & 120000   \\
($\beta$: [1000, 10000])	& HUIs* & 122,531 & 77,739 & 59,193 & 49,645 & 43,466 & 38,930 & 35,198 & 32,066 & 29,395 & 27,065 \\
& HTWUIs &  324,146 & 229,588 & 170,927 & 133,064 & 110,483 & 95,208 & 81,420 & 76,135 & 72,166 & 68,795  \\
& HUIs &   37,459 & 31,989 & 28,721 & 26,601 & 25018 & 23,684 & 22,434 & 21,306 & 20,212 & 19,151   \\	    
\hline
		
		\hline
	\end{tabular}
\end{table}
%%%%%%%%%%%%%%%%%%%%%%%%%%%%%%%%%%%%%%%%%%%%%%%%%%%%%%

%%%%%%%%%%%%%%%%%%%%%%%%%%%%%%%%%%%%%%%%%%%%%%%%%%Friedman
\begin{table}[htb]
	\fontsize{6.5pt}{10pt}\selectfont
	\centering
	\caption{Derived patterns under varied $ \beta $}
	\label{table:patterns2}
	\begin{tabular}{c|c|llllllllll}
		%\toprule
		\hline\hline
		\multirow{2}*{\textbf{Dataset}}&
		\multirow{2}*{\textbf{Pattern}}
		&\multicolumn{10}{c}{\textbf{Threshold}}\\
		\cline{3-12}%\morecmidrules
		& &$ \beta_1 $ & $ \beta_2 $ & $ \beta_3 $ & $ \beta_4 $ &  $ \beta_5 $ &  $ \beta_6 $ &$ \beta_7 $ & $ \beta_8 $ & $ \beta_9 $ & $ \beta_{10} $ \\ \hline
		
		\textbf{foodmart} & $\beta$ &  10 & 20 & 30 & 40 & 50 & 60 & 70 & 80 & 90 & 100 \\
		(\textit{GLMU}: 5k)	& HUIs* &  184,676 & 184,676 & 184,676 & 184,676 & 184,676 & 184,676 & 184,676 & 184,676 & 184,676 & 184676  \\
		& HTWUIs & 480,157 & 470,761 & 383,342 & 266,472 & 188,973 & 126,822 & 81,449 & 48,899 & 31,054 & 20,050   \\
		& HUIs &  184,633 & 170,800 & 107,733 & 66,202 & 42,688 & 23,993 & 12,089 & 6,492 & 3,985 & 2,701  \\	    
		\hline
		
		\textbf{retail} & $\beta$ & 1000 & 2000 & 000 & 4000 & 5000 & 6000 & 7000 & 8000 & 9000 & 10000 \\
		(\textit{GLMU}: 95k)	& HUIs* &  6,546 & 6,546 & 6,546 & 6,546 & 6,546 & 6,546 & 6,546 & 6,546 & 6,546 & 6,546  \\
		& HTWUIs &  518,172 & 453,955 & 281,644 & 242,897 & 183,433 & 151,859 & 131,003 & 122,263 & 115,365 & 110,397   \\
		& HUIs &  6,156 & 4,614 & 2,967 & 2,345 & 1,427 & 1,055 & 855 & 722 & 553 & 483  \\	    
		\hline
		
		\textbf{mushroom} & $\beta$ &  1000 & 2000 & 3000 & 4000 & 5000 & 6000 & 7000 & 8000 & 9000 & 10000  \\
		(\textit{GLMU}: 7500k) & HUIs* &  2,975 & 2,975 & 2,975 & 2,975 & 2,975 & 2,975 & 2,975 & 2,975 & 2,975 & 2,975  \\
		& HTWUIs &  384,117 & 384,117 & 384,117 & 384,117 & 384,117 & 384,117 & 384,117 & 384,117 & 384,117 & 384,117  \\
		& HUIs &  2,975 & 2,975 & 2,975 & 2,975 & 2,975 & 2,975 & 2,975 & 2,975 & 2,975 & 2,975   \\
		\hline
		
		\textbf{T10I4D100K} & $\beta$ &  1000 & 2000 & 3000 & 4000 & 	5000 & 	6000 & 	7000 & 	8000 & 	9000 & 	10000   \\
		(\textit{GLMU}: 80k)	& HUIs* &  38,930 & 38,930 & 	38,930 & 	38,930 & 	38,930 & 	38,930 & 	38,930 & 	38,930 & 	38,930 & 	38,930  \\
		& HTWUIs &  181,338 & 	146,405 & 	96,455 & 	96,132 & 83,927 & 	73,884 & 	67,244 & 	62,017 & 	58,398 & 	55,298  \\
		& HUIs &   38,406 & 	32,998 & 	23,399 & 	23,225 & 19,921 & 	17,053 & 	15,214 & 	13,602 & 	12,330 & 	11,213   \\	    
		\hline
		
		\hline
		%\bottomrule
	\end{tabular}
\end{table}
%%%%%%%%%%%%%%%%%%%%%%%%%%%%%%%%%%%%%%%%%%%%%%%%%%%%%%

From Table \ref{table:patterns1}, it can be observed that the number of HUIs that are derived by the proposed algorithms under various \textit{GLMUs} is always smaller than that of HUIs* whether in sparse or dense datasets. For example, in the foodmart dataset, the number of HUIs* and HUIs are 287,063 and 148,266, respectively, when the \textit{GLMU} is set as 4000. This indicates that numerous HUIs* are discovered using a uniform minimum utility threshold, but few of them are selected as HUIs by considering the distinct minimum utility threshold (\textit{minutil}) of each item in the datasets. In real-world applications, traditional algorithms in HUIM may easily suffer from the ``\textit{rare item problem}"; that is, if the \textit{minutil} is set too high or too low, patterns involving items with high or low utility cannot be found. When the \textit{minutil} is set too low, many meaningless patterns may be found, and this may cause the problem of combinatorial explosion. This situation regularly happened when \textit{GLMU} was set lower. For example, on the foodmart dataset in Table 8, there are 499,063 HUIs*, but only 203,871 are considered as the HUIs when the \textit{GLMU} is set as 2,000; most of the discovered HUIs* may not be interesting in making the useful decisions because most items are treated the same. This shows that the addressed HUIM-MMU can effectively discover fewer but more useful HUIs than the traditional HUIM framework with a uniform minimum utility threshold.

As shown in Table \ref{table:patterns2}, it can be found that the number of HUIs derived by the proposed approaches decreases and is close to the number of HUIs* when $\beta$ is set lower. For example, when the \textit{GLMU} is set from 1,000 to 10,000 on the T10I4D100K dataset, the HUIs are changed from 38,406 to 11,213, but the number of HUIs* remains stable as 38,930. This is reasonable, because when $\beta$ is set lower, the mu value of each item is close to the \textit{GLMU}. When $\beta$ is set higher, the \textit{mu} value of each item is larger than the \textit{GLMU}, and then the number of discovered HUIs will be lesser than that of HUIs*. It can thus be concluded that the ``uniform minimum utility threshold" is a special case of the ``multiple minimum utility thresholds" framework when all thresholds are set to the \textit{GLMU}. Furthermore, an interesting observation is that the number of HUIs is significantly influenced by the minimum utility thresholds. The reason is that the enormous redundant patterns can be reduced, and more meaningful and condensed patterns can be revealed using multiple minimum utility thresholds. It also indicates that some redundant HUIs may not be available in real-world applications, and the proposed HIMU algorithms can effectively avoid the ``\textit{rare item problem}". Hence, the proposed algorithms can dramatically reduce the number of redundant patterns, thus making the task of high-utility itemset mining more realistic in real-life situations.

\subsection{Effect of Pruning Strategies}

We also evaluated the effectiveness of the pruning ability of the EUCP and LAP strategies. Note that the unpruned itemsets w.r.t. the visited nodes by performing the HIMU, HIMU$_{EUCP}$, HIMU$_{LAP}$, and HIMU$_{ELP}$ algorithms are respectively named $N_2$, $N_3$, $N_4$, and $N_5$, respectively. In addition, we also compared the number of generated patterns for deriving HTWUIs in the HUI-MMU and HUI-MMU$_{TID}$ algorithms, denoted $N_1$. The results are shown in Fig. \ref{fig_patterns1} and Fig. \ref{fig_patterns2}, respectively.

\begin{figure*}[hbtp]	
	\centering
	\includegraphics[trim=10 0 5 0,clip,scale=0.38]{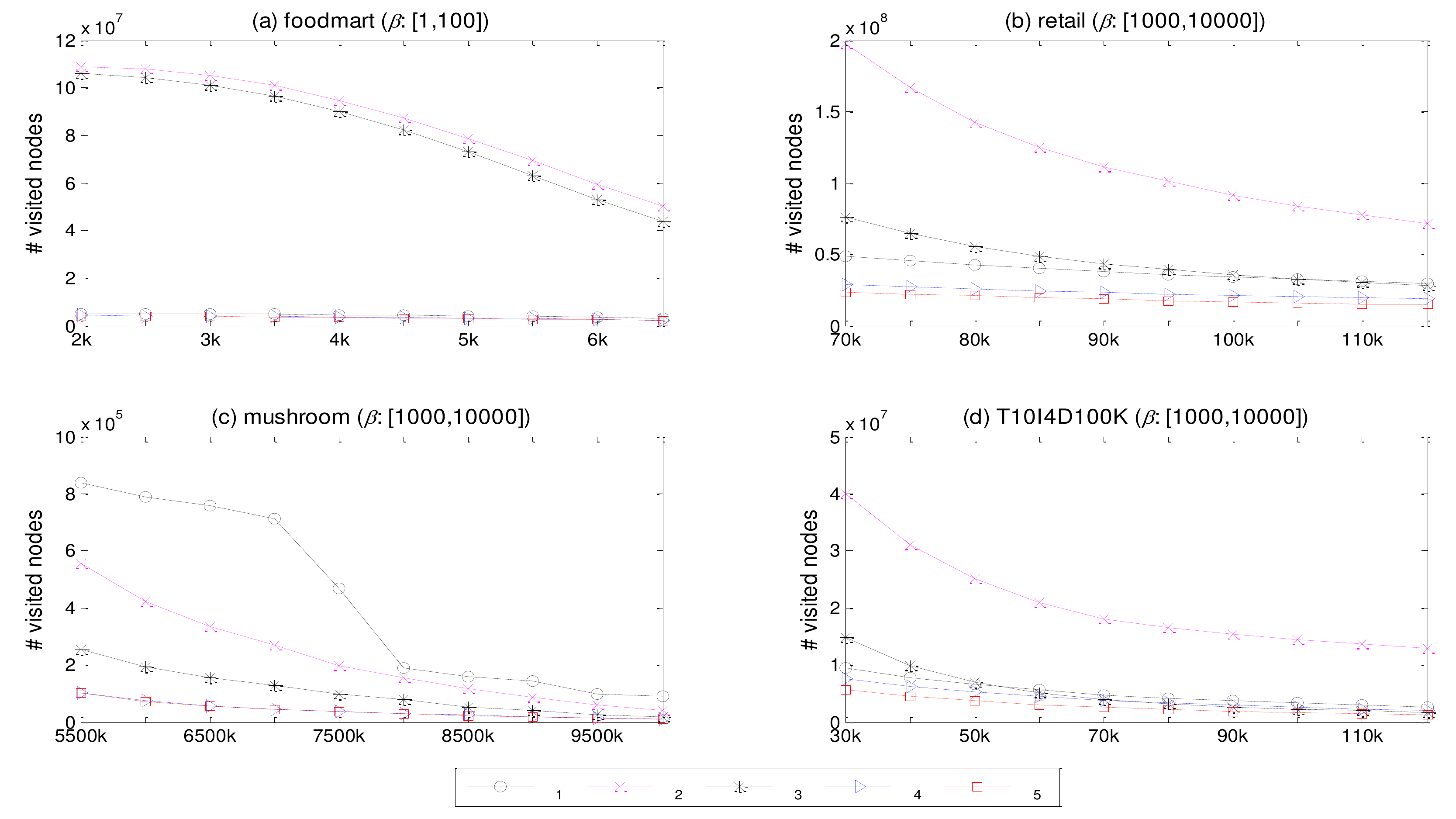}
	\captionsetup{justification=centering}
	\caption{Number of patterns with a fixed $\beta$ under various \textit{GLMUs}.}
	\label{fig_patterns1}
\end{figure*}

\begin{figure*}[hbtp]	
	\centering
	\includegraphics[trim=10 0 5 0,clip,scale=0.38]{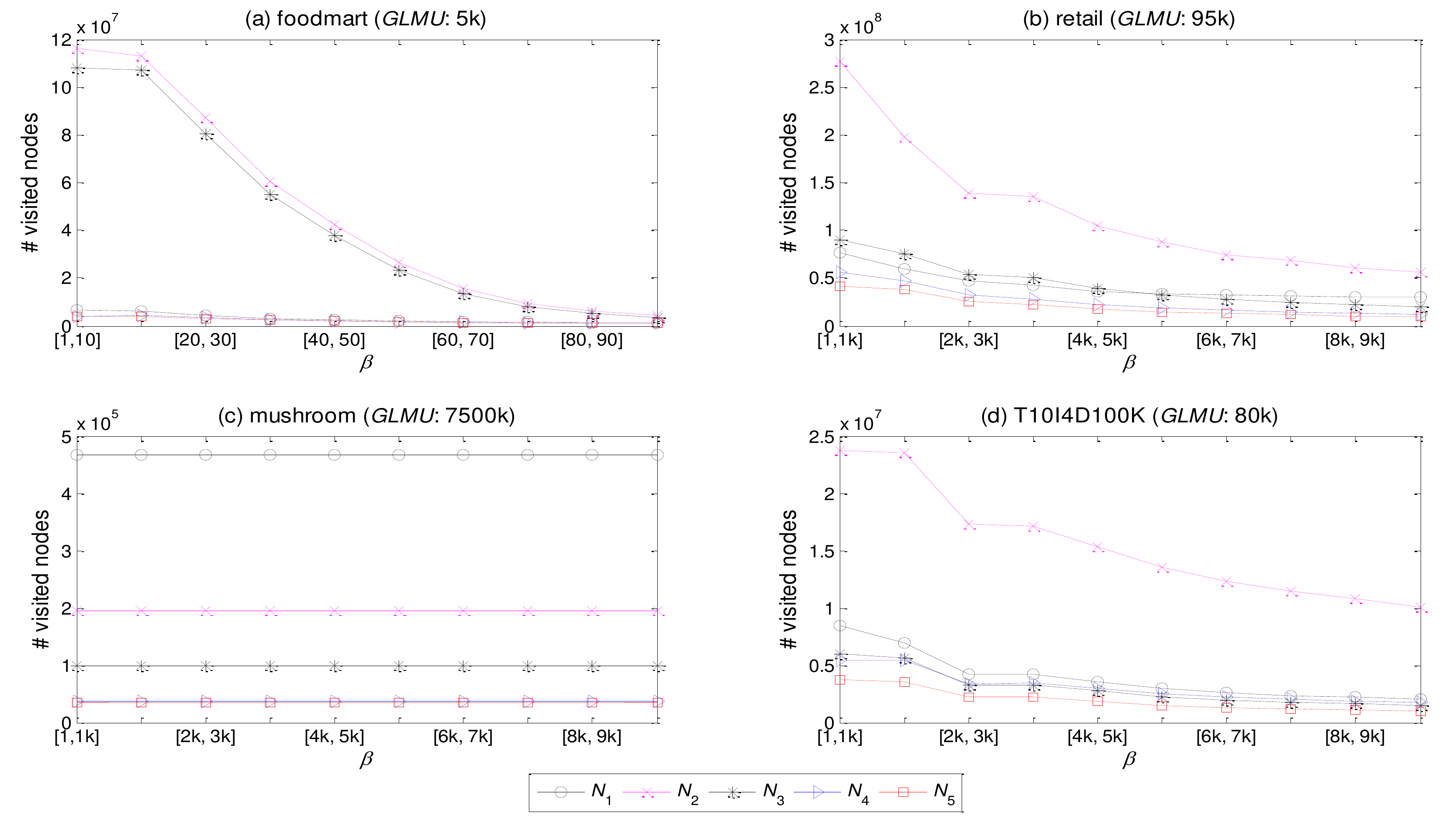}
	\captionsetup{justification=centering}
	\caption{Number of patterns with a fixed \textit{GLMU} under various $\beta$.}
	\label{fig_patterns2}
\end{figure*}

From Fig. \ref{fig_patterns1} and Fig. \ref{fig_patterns2}, it can be observed that the produced itemsets for determining HTWUIs in HUI-MMU and HUI-MMU$_{TID}$, w.r.t. $N_1$, is sometimes less than the visited nodes $N_1$ ($N_2 > N_1$, as shown in the foodmart, retail, and T10I4D100K datasets), but is sometimes greater than $N_1$ [$N_2 \leq N_1$, as shown in Fig. \ref{fig_patterns1}(c) and Fig. \ref{fig_patterns2}(c)]. Moreover, the visited nodes in the proposed HIMU with different pruning strategies always have the following relationship as $ N_2 \geq N_3 \geq N_4 \geq N_5 $ for all datasets. For example, when the \textit{GLMU} and $\beta$ are both set as 90,000 for the retail dataset, the necessary visited nodes of HIMU, HIMU$_{EUCP}$, HIMU$_{LAP}$, and HIMU$_{ELP}$ are, respectively, $N_2$ (= 111,463,003), $N_3$ (= 43,366,282), $N_4$ (= 23,026,436), and $N_5$ (= 18,416,372), as shown in Fig. \ref{fig_patterns1}(b). The reason is that both the EUCP and LAP strategies prune a considerable amount of itemsets early, and thus the generation of their extensions in the proposed enhanced algorithms can be avoided. Among the five types of patterns, the number of $N_5$ is the least, which means that the hybrid method has the best pruning effect. Moreover, it can be seen that $ N_2 \geq N_3 \geq N_1 $ in some cases, as shown in Fig. \ref{fig_patterns1}(a), Fig. \ref{fig_patterns1}(b), Fig. \ref{fig_patterns1}(d), Fig. \ref{fig_patterns2}(a), and Fig. \ref{fig_patterns2}(b). It indicates that the explored space (w.r.t. the visited nodes in the set-enumeration MIU-tree) of HIMU may be very huge without any pruning strategies. It can also be observed that when the \textit{GLMU} or $\beta$ decreases, the gap between the number of these patterns is increased for all datasets whether under various \textit{GLMUs} with a fixed $\beta$ or under various $\beta$ with a fixed \textit{GLMU}.

\subsection{Memory Consumption}

We took the memory measurements using the Java application programming interface (API). Note that we recorded the peak memory consumption of the compared algorithms for all datasets. The results under various \textit{GLMUs} with a fixed $\beta$, and under various $\beta$ with a fixed \textit{GLMU} are shown in Fig. \ref{fig_memory1} and Fig. \ref{fig_memory2}, respectively.

\begin{figure*}[hbtp]	
	\centering
	\includegraphics[trim=10 0 5 0,clip,scale=0.38]{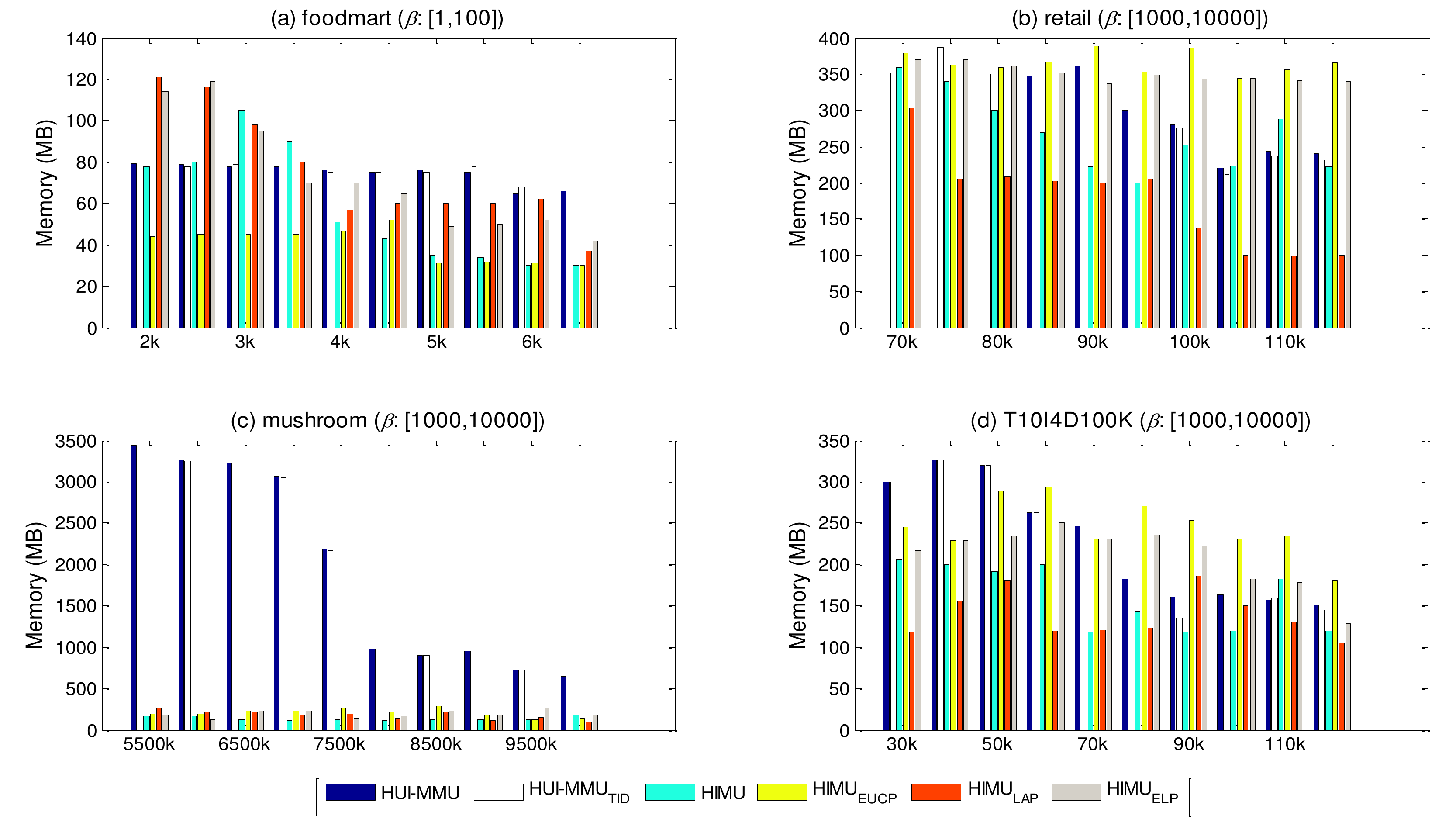}
	\captionsetup{justification=centering}
	\caption{Memory consumption with a fixed $\beta$ under various \textit{GLMUs}.}
	\label{fig_memory1}
\end{figure*}

\begin{figure*}[hbtp]	
	\centering
	\includegraphics[trim=10 0 5 0,clip,scale=0.38]{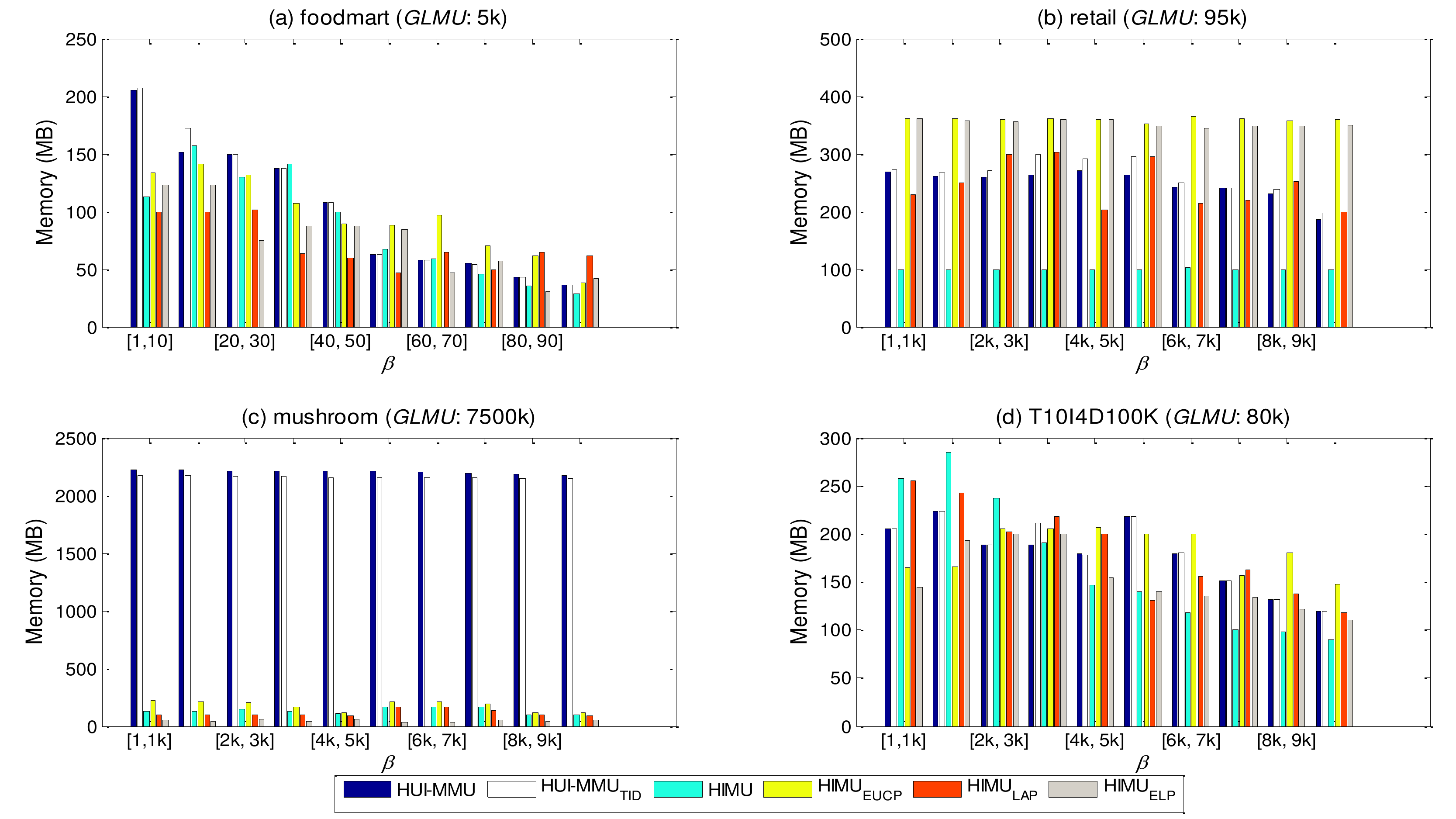}
	\captionsetup{justification=centering}
	\caption{Memory consumption with a fixed \textit{GLMU} under various $\beta$.}
	\label{fig_memory2}
\end{figure*}

From Fig. \ref{fig_memory1} and Fig. \ref{fig_memory2}, it can be clearly seen that the proposed HIMU algorithms with varied pruning strategies require less memory compared with the state-of-the-art HUI-MMU and HUI-MMU$_{TID}$ algorithms, both under various \textit{GLMUs} with a fixed $\beta$ and under various $\beta$ with a fixed \textit{GLMU} for the four datasets. Specifically, the four HIMU algorithms require nearly constant memory under various parameter values on the four datasets. This performance is somehow similar to the HUI-Miner and FHM algorithms, except for the retail dataset. This result is reasonable because the series HIMU algorithms are MIU-tree-based algorithms; they can fast-span the MIU-tree without candidate generation, and the unpromising itemsets can be easily pruned. Furthermore, the utility-list structure is designed as a vertical compact structure to store the necessary information. Less memory is thus consumed. The HIMU$_{EUCP}$ algorithm consumes slightly more memory than the baseline HIMU algorithm because it stores information about the co-occurrence of 2-itemsets in an additional data structure.

Owing to the fact that a depth-first search mechanism is used to explore the MIU-tree, the utility-list structures of the promising itemsets are built during the mining process. Thus, the ``\textit{mining during constructing property}" can greatly reduce the memory consumption, and it does not need to explore the entire search space to generate all candidates. Thanks to the advantages provided by the vertical utility-list structure, the proposed four HIMU algorithms are more efficient to use to discover the HUIs. Hence, the memory consumption of the MIU-tree-based algorithms significantly outperform those of the HUI-MMU and HUI-MMU$_{TID}$ algorithms.

\subsection{Scalability Analysis}

We compared the scalability of the proposed algorithms and other compared algorithms under varied dataset sizes, T10I4N4KD$|X|K$, in terms of runtime, memory usage, the number of generated HUIs and HTWUIs, and the number of visited nodes in MIU-tree. The variable K indicates the dataset size from 100(K) to 500(K), in increments 100(K) each time. When the \textit{GLMU} was set to 1,000,000 and $\beta$ was varied from 1,000 to 10,000, the results of the compared algorithms are shown in Fig. \ref{fig_scalability}. The details of each type of analysis follow.

\begin{figure*}[hbtp]	
	\centering
	\includegraphics[trim=10 0 5 0,clip,scale=0.38]{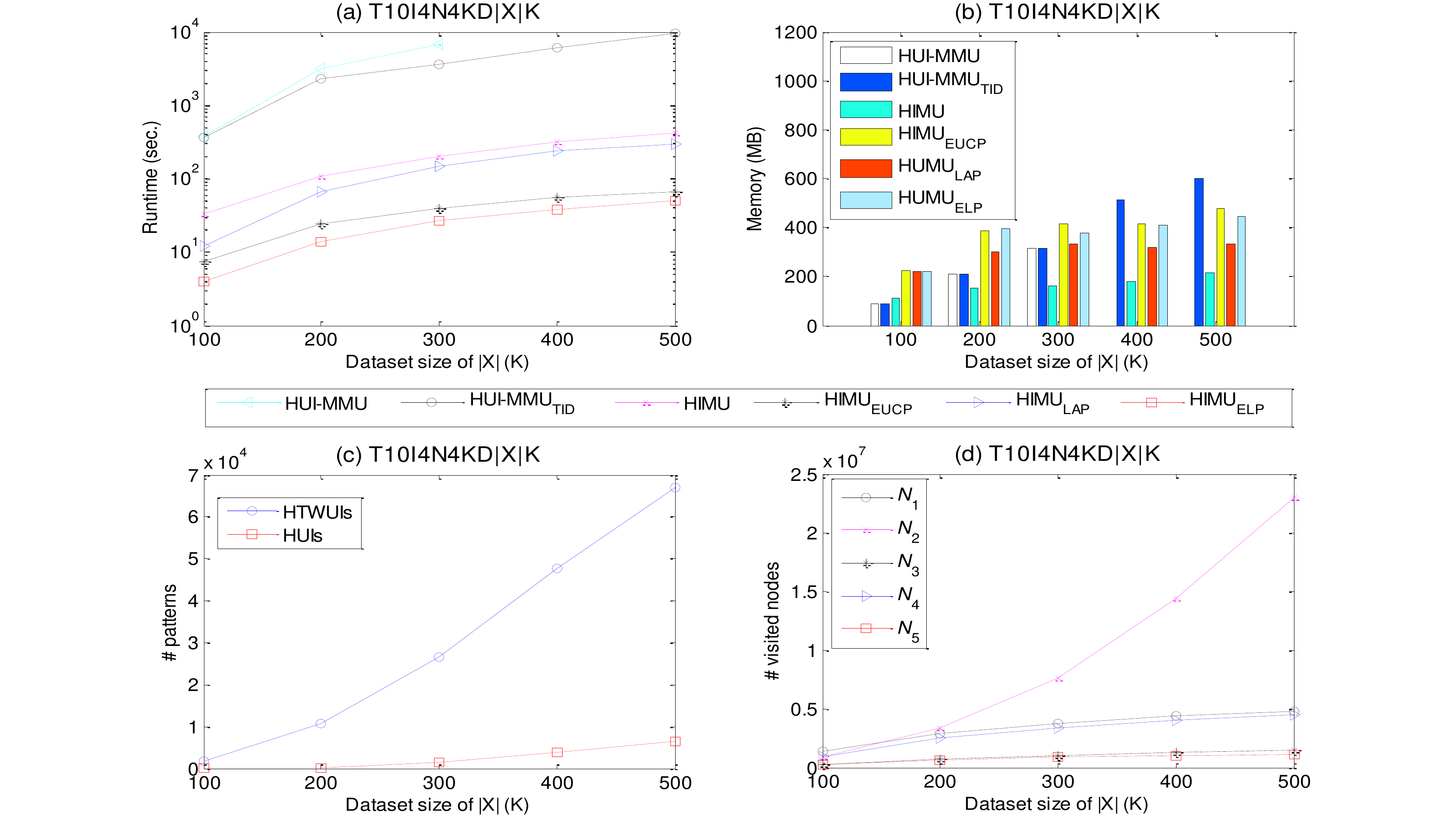}
	\captionsetup{justification=centering}
	\caption{Scalability under varied dataset sizes.}
	\label{fig_scalability}
\end{figure*}

(1) \textit{Runtime analysis}. From Fig. \ref{fig_scalability}(a), it is evident that all of the compared algorithms scale well along with increasing dataset size. When the dataset size increases, the gap between the proposed algorithms and other compared algorithms increases in terms of execution time and memory consumption. Hence, the powerful pruning strategies of the developed HIMU algorithm showed better results compared with the previous algorithms for larger dataset sizes.

(2) \textit{Memory usage analysis}. From Fig. \ref{fig_scalability}(b), it can be noted that the proposed two algorithms with different pruning strategies always consume less memory than the compared algorithms on all datasets. As mentioned before, the advantages of the proposed MIU-tree-based algorithms are that they can easily maintain the necessary information by constructing a series of utility-list structures. Without multiple database scans and a generate-and-test approach of candidates, the memory cost can be greatly reduced. In addition, with the effect of the developed pruning strategies, less memory consumption can thus be obtained. The EUCS was also further designed as an additional structure to keep the TWU value of 2-itemsets with light memory usage. Hence, the HIMU$_{EUCP}$, HIMU$_{LAP}$, and HIMU$_{ELP}$ algorithms consume little more memory than the baseline HIMU algorithm.

(3) \textit{Patterns analysis for HTWUIs and HUIs}. From Fig. \ref{fig_scalability}(c), it is also evident that the number of discovered HUIs is quite less than that of the HTWUIs derived by the HUI-MMU algorithm. Moreover, the larger the dataset size, the larger the gap between the number of HTWUIs and HUIs. Thus, the previous level-wise mining approaches perform worse than the utility-list-based HIMU model in the scalability test.

From these experiments, we find that the proposed improved HIMU$_{EUCP}$, HIMU$_{LAP}$, and HIMU$_{ELP}$ algorithms outperform the baseline HIMU algorithm, as well as the state-of-the-art HUI-MMU and HUI-MMU$_{TID}$ algorithms, in terms of runtime, memory usage, and number of HUIs, along with the increase of dataset size. Moreover, the improved algorithms are faster than the baseline algorithm, with lighter memory usage as well.

\section{Conclusion} %  and Future Works
\label{sec:conclusion}

Consumer behavior plays a very important role in economics and targeted marketing. However, understanding economic consumer behavior is quite challenging. Moreover, in some real-world applications, some items have hierarchies and unit contributed utility can differ. In this paper, we designed a novel set-enumeration tree-based algorithm named HIMU to discover high-utility itemsets with multiple minimum utility thresholds. This is the first work, to our knowledge, to propose efficient one-phase algorithms to address the HUI-MMU problem. A compact multiple item utility set-enumeration tree (MIU-tree) is designed for directly mining HUIs by spanning the MIU-tree without a generation-and-test approach. In addition, we propose the \textit{global} and \textit{conditional downward closure} (\textit{GDC} and \textit{CDC}) properties to guarantee global and partial anti-monotonicity for HUIs. The necessary information of itemsets from the processed databases can be easily obtained from a series of compact utility-list structures of their prefix itemsets in the designed MIU-tree. The HIMU algorithm thus can directly discover the HUIs without candidate generation or multiple database scans. Moreover, two enhanced EUCP and LAP pruning strategies are applied in the improved HIMU$_{EUCP}$, HIMU$_{LAP}$, and HIMU$_{ELP}$ algorithms, thus speeding up the mining performance.

From experiments conducted on both real-life and synthetic datasets, it can be observed that the four proposed algorithms have better performance in revealing the meaningful HUIs from the databases with multiple minimum utility thresholds and that they not only avoid the ``\textit{rare item problem}" that easily occurs in the traditional HUIM algorithms but also are more efficient than the state-of-the-art HUI-MMU and HUI-MMU$_{TID}$ algorithms in terms of runtime, memory usage, and scalability.

%%%%%%%%%%%%%%%%%%%%%%%%%%%%%%%%%%%%%%%%%%%%%%%%%%%%%%%%%%%%%%%%%%%%%%%%%%

\bibliographystyle{ACM-Reference-Format}
\bibliography{acmlarge}


\begin{thebibliography}{63}

%%% the Web address.

\ifx \showCODEN    \undefined \def \showCODEN     #1{\unskip}     \fi
\ifx \showDOI      \undefined \def \showDOI       #1{#1}\fi
\ifx \showISBNx    \undefined \def \showISBNx     #1{\unskip}     \fi
\ifx \showISBNxiii \undefined \def \showISBNxiii  #1{\unskip}     \fi
\ifx \showISSN     \undefined \def \showISSN      #1{\unskip}     \fi
\ifx \showLCCN     \undefined \def \showLCCN      #1{\unskip}     \fi
\ifx \shownote     \undefined \def \shownote      #1{#1}          \fi
\ifx \showarticletitle \undefined \def \showarticletitle #1{#1}   \fi
\ifx \showURL      \undefined \def \showURL       {\relax}        \fi
% invisible to TeX
\providecommand\bibfield[2]{#2}
\providecommand\bibinfo[2]{#2}
\providecommand\natexlab[1]{#1}
\providecommand\showeprint[2][]{arXiv:#2}

\bibitem[\protect\citeauthoryear{Agrawal, Imieli{\'n}ski, and Swami}{Agrawal
  et~al\mbox{.}}{1993}]%
        {agrawal1993mining}
\bibfield{author}{\bibinfo{person}{Rakesh Agrawal}, \bibinfo{person}{Tomasz
  Imieli{\'n}ski}, {and} \bibinfo{person}{Arun Swami}.}
  \bibinfo{year}{1993}\natexlab{}.
\newblock \showarticletitle{Mining association rules between sets of items in
  large databases}. In \bibinfo{booktitle}{\emph{Acm sigmod record}},
  Vol.~\bibinfo{volume}{22}. ACM, \bibinfo{pages}{207--216}.
\newblock


\bibitem[\protect\citeauthoryear{Agrawal and Srikant}{Agrawal and
  Srikant}{1995}]%
        {agrawal1995mining}
\bibfield{author}{\bibinfo{person}{Rakesh Agrawal} {and}
  \bibinfo{person}{Ramakrishnan Srikant}.} \bibinfo{year}{1995}\natexlab{}.
\newblock \showarticletitle{Mining sequential patterns}. In
  \bibinfo{booktitle}{\emph{Proceedings of the Eleventh International
  Conference on Data Engineering}}. IEEE, \bibinfo{pages}{3--14}.
\newblock


\bibitem[\protect\citeauthoryear{Agrawal, Srikant, et~al\mbox{.}}{Agrawal
  et~al\mbox{.}}{1994}]%
        {agrawal1994fast}
\bibfield{author}{\bibinfo{person}{Rakesh Agrawal},
  \bibinfo{person}{Ramakrishnan Srikant}, {et~al\mbox{.}}}
  \bibinfo{year}{1994}\natexlab{}.
\newblock \showarticletitle{Fast algorithms for mining association rules}. In
  \bibinfo{booktitle}{\emph{Proceedings of the 20th International Conference on
  Very Large Data Bases}}, Vol.~\bibinfo{volume}{1215}.
  \bibinfo{pages}{487--499}.
\newblock


\bibitem[\protect\citeauthoryear{Agrawal and Srikant}{Agrawal and
  Srikant}{1994}]%
        {IBMdata}
\bibfield{author}{\bibinfo{person}{Rakesh Agrawal} {and}
  \bibinfo{person}{Ramakrishnan. Quest synthetic data~generator Srikant}.}
  \bibinfo{year}{1994}\natexlab{}.
\newblock
  \bibinfo{title}{\url{http://www.Almaden.ibm.com/cs/quest/syndata.html}}.
\newblock   (\bibinfo{year}{1994}).
\newblock


\bibitem[\protect\citeauthoryear{Ahmed, Tanbeer, Jeong, and Lee}{Ahmed
  et~al\mbox{.}}{2009}]%
        {ahmed2009efficient}
\bibfield{author}{\bibinfo{person}{Chowdhury~Farhan Ahmed},
  \bibinfo{person}{Syed~Khairuzzaman Tanbeer}, \bibinfo{person}{Byeong-Soo
  Jeong}, {and} \bibinfo{person}{Young-Koo Lee}.}
  \bibinfo{year}{2009}\natexlab{}.
\newblock \showarticletitle{Efficient tree structures for high utility pattern
  mining in incremental databases}.
\newblock \bibinfo{journal}{\emph{IEEE Transactions on Knowledge and Data
  Engineering}} \bibinfo{volume}{21}, \bibinfo{number}{12}
  (\bibinfo{year}{2009}), \bibinfo{pages}{1708--1721}.
\newblock


\bibitem[\protect\citeauthoryear{Alkan and Karagoz}{Alkan and Karagoz}{2015}]%
        {alkan2015crom}
\bibfield{author}{\bibinfo{person}{Oznur~Kirmemis Alkan} {and}
  \bibinfo{person}{Pinar Karagoz}.} \bibinfo{year}{2015}\natexlab{}.
\newblock \showarticletitle{{CROM} and {H}usp{E}xt: Improving efficiency of
  high utility sequential pattern extraction}.
\newblock \bibinfo{journal}{\emph{IEEE Transactions on Knowledge and Data
  Engineering}} \bibinfo{volume}{27}, \bibinfo{number}{10}
  (\bibinfo{year}{2015}), \bibinfo{pages}{2645--2657}.
\newblock


\bibitem[\protect\citeauthoryear{Chan, Yang, and Shen}{Chan
  et~al\mbox{.}}{2003}]%
        {chan2003mining}
\bibfield{author}{\bibinfo{person}{Raymond Chan}, \bibinfo{person}{Qiang Yang},
  {and} \bibinfo{person}{Yi-Dong Shen}.} \bibinfo{year}{2003}\natexlab{}.
\newblock \showarticletitle{Mining high utility itemsets}. In
  \bibinfo{booktitle}{\emph{Third IEEE International Conference on Data
  Mining}}. IEEE, \bibinfo{pages}{19--26}.
\newblock


\bibitem[\protect\citeauthoryear{Chen, Han, and Yu}{Chen et~al\mbox{.}}{1996}]%
        {chen1996data}
\bibfield{author}{\bibinfo{person}{Ming-Syan Chen}, \bibinfo{person}{Jiawei
  Han}, {and} \bibinfo{person}{Philip~S. Yu}.} \bibinfo{year}{1996}\natexlab{}.
\newblock \showarticletitle{Data mining: an overview from a database
  perspective}.
\newblock \bibinfo{journal}{\emph{IEEE Transactions on Knowledge and data
  Engineering}} \bibinfo{volume}{8}, \bibinfo{number}{6}
  (\bibinfo{year}{1996}), \bibinfo{pages}{866--883}.
\newblock


\bibitem[\protect\citeauthoryear{database foodmart of microsoft~analysis
  services}{database foodmart of microsoft~analysis services}{[n. d.]}]%
        {foodmart}
\bibfield{author}{\bibinfo{person}{Microsoft.~Example database foodmart of
  microsoft~analysis services}.} \bibinfo{year}{[n. d.]}\natexlab{}.
\newblock
  \bibinfo{title}{\url{http://msdn.microsoft.com/en-us/library/aa217032(SQL.80).aspx}}.
\newblock   (\bibinfo{year}{[n. d.]}).
\newblock


\bibitem[\protect\citeauthoryear{Fournier-Viger, Faghihi, Nkambou, and
  Nguifo}{Fournier-Viger et~al\mbox{.}}{2012}]%
        {fournier2012cmrules}
\bibfield{author}{\bibinfo{person}{Philippe Fournier-Viger},
  \bibinfo{person}{Usef Faghihi}, \bibinfo{person}{Roger Nkambou}, {and}
  \bibinfo{person}{Engelbert~Mephu Nguifo}.} \bibinfo{year}{2012}\natexlab{}.
\newblock \showarticletitle{{CMR}ules: Mining sequential rules common to
  several sequences}.
\newblock \bibinfo{journal}{\emph{Knowledge-Based Systems}}
  \bibinfo{volume}{25}, \bibinfo{number}{1} (\bibinfo{year}{2012}),
  \bibinfo{pages}{63--76}.
\newblock


\bibitem[\protect\citeauthoryear{Fournier-Viger, Lin, Gomariz, Gueniche,
  Soltani, Deng, and Lam}{Fournier-Viger et~al\mbox{.}}{2016}]%
        {fournier2016spmf}
\bibfield{author}{\bibinfo{person}{Philippe Fournier-Viger},
  \bibinfo{person}{Jerry Chun-Wei Lin}, \bibinfo{person}{Antonio Gomariz},
  \bibinfo{person}{Ted Gueniche}, \bibinfo{person}{Azadeh Soltani},
  \bibinfo{person}{Zhihong Deng}, {and} \bibinfo{person}{Hoang~Thanh Lam}.}
  \bibinfo{year}{2016}\natexlab{}.
\newblock \showarticletitle{The SPMF open-source data mining library version
  2}. In \bibinfo{booktitle}{\emph{Joint European Conference on Machine
  Learning and Knowledge Discovery in Databases}}. Springer,
  \bibinfo{pages}{36--40}.
\newblock


\bibitem[\protect\citeauthoryear{Fournier-Viger, Wu, Zida, and
  Tseng}{Fournier-Viger et~al\mbox{.}}{2014}]%
        {fournier2014fhm}
\bibfield{author}{\bibinfo{person}{Philippe Fournier-Viger},
  \bibinfo{person}{Cheng-Wei Wu}, \bibinfo{person}{Souleymane Zida}, {and}
  \bibinfo{person}{Vincent~S Tseng}.} \bibinfo{year}{2014}\natexlab{}.
\newblock \showarticletitle{FHM: Faster high-utility itemset mining using
  estimated utility co-occurrence pruning}. In
  \bibinfo{booktitle}{\emph{International Symposium on Methodologies for
  Intelligent Systems}}. Springer, \bibinfo{pages}{83--92}.
\newblock


\bibitem[\protect\citeauthoryear{Gan, Lin, Chao, Wang, and Philip}{Gan
  et~al\mbox{.}}{2018a}]%
        {gan2018privacy}
\bibfield{author}{\bibinfo{person}{Wensheng Gan}, \bibinfo{person}{Jerry
  Chun-Wei Lin}, \bibinfo{person}{Han-Chieh Chao}, \bibinfo{person}{Shyue-Liang
  Wang}, {and} \bibinfo{person}{S~Yu Philip}.}
  \bibinfo{year}{2018}\natexlab{a}.
\newblock \showarticletitle{Privacy preserving utility mining: a survey}. In
  \bibinfo{booktitle}{\emph{IEEE International Conference on Big Data}}. IEEE,
  \bibinfo{pages}{2617--2626}.
\newblock


\bibitem[\protect\citeauthoryear{Gan, Lin, Chao, and Zhan}{Gan
  et~al\mbox{.}}{2017a}]%
        {gan2017data}
\bibfield{author}{\bibinfo{person}{Wensheng Gan}, \bibinfo{person}{Jerry
  Chun-Wei Lin}, \bibinfo{person}{Han-Chieh Chao}, {and}
  \bibinfo{person}{Justin Zhan}.} \bibinfo{year}{2017}\natexlab{a}.
\newblock \showarticletitle{Data mining in distributed environment: a survey}.
\newblock \bibinfo{journal}{\emph{Wiley Interdisciplinary Reviews: Data Mining
  and Knowledge Discovery}} \bibinfo{volume}{7}, \bibinfo{number}{6}
  (\bibinfo{year}{2017}).
\newblock


\bibitem[\protect\citeauthoryear{Gan, Lin, Fournier-Viger, Chao, and
  Fujita}{Gan et~al\mbox{.}}{2018b}]%
        {gan2017extracting}
\bibfield{author}{\bibinfo{person}{Wensheng Gan}, \bibinfo{person}{Jerry
  Chun-Wei Lin}, \bibinfo{person}{Philippe Fournier-Viger},
  \bibinfo{person}{Han-Chieh Chao}, {and} \bibinfo{person}{Hamido Fujita}.}
  \bibinfo{year}{2018}\natexlab{b}.
\newblock \showarticletitle{Extracting non-redundant correlated purchase
  behaviors by utility measure}.
\newblock \bibinfo{journal}{\emph{Knowledge-Based Systems}}
  \bibinfo{volume}{143} (\bibinfo{year}{2018}), \bibinfo{pages}{30--41}.
\newblock


\bibitem[\protect\citeauthoryear{Gan, Lin, Fournier-Viger, Chao, Hong, and
  Fujita}{Gan et~al\mbox{.}}{2018c}]%
        {gan2018survey}
\bibfield{author}{\bibinfo{person}{Wensheng Gan}, \bibinfo{person}{Jerry
  Chun-Wei Lin}, \bibinfo{person}{Philippe Fournier-Viger},
  \bibinfo{person}{Han-Chieh Chao}, \bibinfo{person}{Tzung-Pei Hong}, {and}
  \bibinfo{person}{Hamido Fujita}.} \bibinfo{year}{2018}\natexlab{c}.
\newblock \showarticletitle{A survey of incremental high-utility itemset
  mining}.
\newblock \bibinfo{journal}{\emph{Wiley Interdisciplinary Reviews: Data Mining
  and Knowledge Discovery}} \bibinfo{volume}{8}, \bibinfo{number}{2}
  (\bibinfo{year}{2018}).
\newblock


\bibitem[\protect\citeauthoryear{Gan, Lin, Fournier-Viger, Chao, and Tseng}{Gan
  et~al\mbox{.}}{2017b}]%
        {2gan2017mining}
\bibfield{author}{\bibinfo{person}{Wensheng Gan}, \bibinfo{person}{Jerry
  Chun-Wei Lin}, \bibinfo{person}{Philippe Fournier-Viger},
  \bibinfo{person}{Han-Chieh Chao}, {and} \bibinfo{person}{Vincent~S Tseng}.}
  \bibinfo{year}{2017}\natexlab{b}.
\newblock \showarticletitle{Mining high-utility itemsets with both positive and
  negative unit profits from uncertain databases}. In
  \bibinfo{booktitle}{\emph{Pacific-Asia Conference on Knowledge Discovery and
  Data Mining}}. Springer, \bibinfo{pages}{434--446}.
\newblock


\bibitem[\protect\citeauthoryear{Gan, Lin, Fournier-Viger, Chao, Tseng, and
  Yu}{Gan et~al\mbox{.}}{2018d}]%
        {2gan2018survey}
\bibfield{author}{\bibinfo{person}{Wensheng Gan}, \bibinfo{person}{Jerry
  Chun-Wei Lin}, \bibinfo{person}{Philippe Fournier-Viger},
  \bibinfo{person}{Han-Chieh Chao}, \bibinfo{person}{Vincent~S Tseng}, {and}
  \bibinfo{person}{Philip~S Yu}.} \bibinfo{year}{2018}\natexlab{d}.
\newblock \showarticletitle{A survey of utility-oriented pattern mining}.
\newblock \bibinfo{journal}{\emph{arXiv preprint arXiv:1805.10511}}
  (\bibinfo{year}{2018}).
\newblock


\bibitem[\protect\citeauthoryear{Gan, Lin, Fournier-Viger, Chao, and Yu}{Gan
  et~al\mbox{.}}{2019}]%
        {gan2018huopm}
\bibfield{author}{\bibinfo{person}{Wensheng Gan}, \bibinfo{person}{Jerry
  Chun-Wei Lin}, \bibinfo{person}{Philippe Fournier-Viger},
  \bibinfo{person}{Han-Chieh Chao}, {and} \bibinfo{person}{Philip~S Yu}.}
  \bibinfo{year}{2019}\natexlab{}.
\newblock \showarticletitle{{HUOPM}: High utility occupancy pattern mining}.
\newblock \bibinfo{journal}{\emph{IEEE Transactions on Cybernetics}}
  \bibinfo{number}{10.1109/TCYB.2019.2896267} (\bibinfo{year}{2019}).
\newblock


\bibitem[\protect\citeauthoryear{Gan, Lin, Fournier-Viger, Chao, and Zhan}{Gan
  et~al\mbox{.}}{2017c}]%
        {gan2017mining}
\bibfield{author}{\bibinfo{person}{Wensheng Gan}, \bibinfo{person}{Jerry
  Chun-Wei Lin}, \bibinfo{person}{Philippe Fournier-Viger},
  \bibinfo{person}{Han-Chieh Chao}, {and} \bibinfo{person}{Justin Zhan}.}
  \bibinfo{year}{2017}\natexlab{c}.
\newblock \showarticletitle{Mining of frequent patterns with multiple minimum
  supports}.
\newblock \bibinfo{journal}{\emph{Engineering Applications of Artificial
  Intelligence}}  \bibinfo{volume}{60} (\bibinfo{year}{2017}),
  \bibinfo{pages}{83--96}.
\newblock


\bibitem[\protect\citeauthoryear{Geng and Hamilton}{Geng and Hamilton}{2006}]%
        {geng2006interestingness}
\bibfield{author}{\bibinfo{person}{Liqiang Geng} {and}
  \bibinfo{person}{Howard~J Hamilton}.} \bibinfo{year}{2006}\natexlab{}.
\newblock \showarticletitle{Interestingness measures for data mining: A
  survey}.
\newblock \bibinfo{journal}{\emph{Comput. Surveys}} \bibinfo{volume}{38},
  \bibinfo{number}{3} (\bibinfo{year}{2006}), \bibinfo{pages}{9}.
\newblock


\bibitem[\protect\citeauthoryear{Han, Pei, Yin, and Mao}{Han
  et~al\mbox{.}}{2004}]%
        {han2004mining}
\bibfield{author}{\bibinfo{person}{Jiawei Han}, \bibinfo{person}{Jian Pei},
  \bibinfo{person}{Yiwen Yin}, {and} \bibinfo{person}{Runying Mao}.}
  \bibinfo{year}{2004}\natexlab{}.
\newblock \showarticletitle{Mining frequent patterns without candidate
  generation: A frequent-pattern tree approach}.
\newblock \bibinfo{journal}{\emph{Data Mining and Knowledge Discovery}}
  \bibinfo{volume}{8}, \bibinfo{number}{1} (\bibinfo{year}{2004}),
  \bibinfo{pages}{53--87}.
\newblock


\bibitem[\protect\citeauthoryear{Hong, Wu, and Wang}{Hong
  et~al\mbox{.}}{2009}]%
        {hong2009effective}
\bibfield{author}{\bibinfo{person}{Tzung-Pei Hong}, \bibinfo{person}{Yi-Ying
  Wu}, {and} \bibinfo{person}{Shyue-Liang Wang}.}
  \bibinfo{year}{2009}\natexlab{}.
\newblock \showarticletitle{An effective mining approach for up-to-date
  patterns}.
\newblock \bibinfo{journal}{\emph{Expert Systems with Applications}}
  \bibinfo{volume}{36}, \bibinfo{number}{6} (\bibinfo{year}{2009}),
  \bibinfo{pages}{9747--9752}.
\newblock


\bibitem[\protect\citeauthoryear{Hu and Chen}{Hu and Chen}{2006}]%
        {hu2006mining}
\bibfield{author}{\bibinfo{person}{Ya-Han Hu} {and} \bibinfo{person}{Yen-Liang
  Chen}.} \bibinfo{year}{2006}\natexlab{}.
\newblock \showarticletitle{Mining association rules with multiple minimum
  supports: a new mining algorithm and a support tuning mechanism}.
\newblock \bibinfo{journal}{\emph{Decision Support Systems}}
  \bibinfo{volume}{42}, \bibinfo{number}{1} (\bibinfo{year}{2006}),
  \bibinfo{pages}{1--24}.
\newblock


\bibitem[\protect\citeauthoryear{Huang}{Huang}{2013}]%
        {huang2013discovery}
\bibfield{author}{\bibinfo{person}{Tony Cheng-Kui Huang}.}
  \bibinfo{year}{2013}\natexlab{}.
\newblock \showarticletitle{Discovery of fuzzy quantitative sequential patterns
  with multiple minimum supports and adjustable membership functions}.
\newblock \bibinfo{journal}{\emph{Information Sciences}}  \bibinfo{volume}{222}
  (\bibinfo{year}{2013}), \bibinfo{pages}{126--146}.
\newblock


\bibitem[\protect\citeauthoryear{itemset mining~dataset repository}{itemset
  mining~dataset repository}{2012}]%
        {fimdatasets}
\bibfield{author}{\bibinfo{person}{Frequent itemset mining~dataset
  repository}.} \bibinfo{year}{2012}\natexlab{}.
\newblock \showarticletitle{\url{http://fimi.ua.ac.be/data/}}.
\newblock


\bibitem[\protect\citeauthoryear{Kiran and Reddy}{Kiran and Reddy}{2011}]%
        {kiran2011novel}
\bibfield{author}{\bibinfo{person}{R~Uday Kiran} {and}
  \bibinfo{person}{P~Krishna Reddy}.} \bibinfo{year}{2011}\natexlab{}.
\newblock \showarticletitle{Novel techniques to reduce search space in multiple
  minimum supports-based frequent pattern mining algorithms}. In
  \bibinfo{booktitle}{\emph{Proceedings of the 14th International Conference on
  Extending Database Technology}}. ACM, \bibinfo{pages}{11--20}.
\newblock


\bibitem[\protect\citeauthoryear{Krishnamoorthy}{Krishnamoorthy}{2015}]%
        {krishnamoorthy2015pruning}
\bibfield{author}{\bibinfo{person}{Srikumar Krishnamoorthy}.}
  \bibinfo{year}{2015}\natexlab{}.
\newblock \showarticletitle{Pruning strategies for mining high utility
  itemsets}.
\newblock \bibinfo{journal}{\emph{Expert Systems with Applications}}
  \bibinfo{volume}{42}, \bibinfo{number}{5} (\bibinfo{year}{2015}),
  \bibinfo{pages}{2371--2381}.
\newblock


\bibitem[\protect\citeauthoryear{Lan, Hong, and Tseng}{Lan
  et~al\mbox{.}}{2014a}]%
        {lan2014efficient}
\bibfield{author}{\bibinfo{person}{Guo-Cheng Lan}, \bibinfo{person}{Tzung-Pei
  Hong}, {and} \bibinfo{person}{Vincent~S Tseng}.}
  \bibinfo{year}{2014}\natexlab{a}.
\newblock \showarticletitle{An efficient projection-based indexing approach for
  mining high utility itemsets}.
\newblock \bibinfo{journal}{\emph{Knowledge and Information Systems}}
  \bibinfo{volume}{38}, \bibinfo{number}{1} (\bibinfo{year}{2014}),
  \bibinfo{pages}{85--107}.
\newblock


\bibitem[\protect\citeauthoryear{Lan, Hong, Tseng, and Wang}{Lan
  et~al\mbox{.}}{2014b}]%
        {lan2014applying}
\bibfield{author}{\bibinfo{person}{Guo-Cheng Lan}, \bibinfo{person}{Tzung-Pei
  Hong}, \bibinfo{person}{Vincent~S Tseng}, {and} \bibinfo{person}{Shyue-Liang
  Wang}.} \bibinfo{year}{2014}\natexlab{b}.
\newblock \showarticletitle{Applying the maximum utility measure in high
  utility sequential pattern mining}.
\newblock \bibinfo{journal}{\emph{Expert Systems with Applications}}
  \bibinfo{volume}{41}, \bibinfo{number}{11} (\bibinfo{year}{2014}),
  \bibinfo{pages}{5071--5081}.
\newblock


\bibitem[\protect\citeauthoryear{Lee, Hong, and Lin}{Lee et~al\mbox{.}}{2004}]%
        {lee2004mining}
\bibfield{author}{\bibinfo{person}{Yeong-Chyi Lee}, \bibinfo{person}{Tzung-Pei
  Hong}, {and} \bibinfo{person}{Wen-Yang Lin}.}
  \bibinfo{year}{2004}\natexlab{}.
\newblock \showarticletitle{Mining fuzzy association rules with multiple
  minimum supports using maximum constraints}. In
  \bibinfo{booktitle}{\emph{International Conference on Knowledge-Based and
  Intelligent Information and Engineering Systems}}. Springer,
  \bibinfo{pages}{1283--1290}.
\newblock


\bibitem[\protect\citeauthoryear{Lee, Hong, and Wang}{Lee
  et~al\mbox{.}}{2006}]%
        {lee2006mining}
\bibfield{author}{\bibinfo{person}{Yeong-Chyi Lee}, \bibinfo{person}{Tzung-Pei
  Hong}, {and} \bibinfo{person}{Tien-Chin Wang}.}
  \bibinfo{year}{2006}\natexlab{}.
\newblock \showarticletitle{Mining fuzzy multiple-level association rules under
  multiple minimum supports}. In \bibinfo{booktitle}{\emph{IEEE International
  Conference on Systems, Man and Cybernetics}}, Vol.~\bibinfo{volume}{5}. IEEE,
  \bibinfo{pages}{4112--4117}.
\newblock


\bibitem[\protect\citeauthoryear{Li, Ghose, and Ipeirotis}{Li
  et~al\mbox{.}}{2011}]%
        {li2011towards}
\bibfield{author}{\bibinfo{person}{Beibei Li}, \bibinfo{person}{Anindya Ghose},
  {and} \bibinfo{person}{Panagiotis~G Ipeirotis}.}
  \bibinfo{year}{2011}\natexlab{}.
\newblock \showarticletitle{Towards a theory model for product search}. In
  \bibinfo{booktitle}{\emph{Proceedings of the 20th international conference on
  World wide web}}. ACM, \bibinfo{pages}{327--336}.
\newblock


\bibitem[\protect\citeauthoryear{Lin, Hong, and Lu}{Lin et~al\mbox{.}}{2011}]%
        {lin2011effective}
\bibfield{author}{\bibinfo{person}{Chun-Wei Lin}, \bibinfo{person}{Tzung-Pei
  Hong}, {and} \bibinfo{person}{Wen-Hsiang Lu}.}
  \bibinfo{year}{2011}\natexlab{}.
\newblock \showarticletitle{An effective tree structure for mining high utility
  itemsets}.
\newblock \bibinfo{journal}{\emph{Expert Systems with Applications}}
  \bibinfo{volume}{38}, \bibinfo{number}{6} (\bibinfo{year}{2011}),
  \bibinfo{pages}{7419--7424}.
\newblock


\bibitem[\protect\citeauthoryear{Lin, Lan, and Hong}{Lin
  et~al\mbox{.}}{2015b}]%
        {2lin2015mining}
\bibfield{author}{\bibinfo{person}{Chun-Wei Lin}, \bibinfo{person}{Guo-Cheng
  Lan}, {and} \bibinfo{person}{Tzung-Pei Hong}.}
  \bibinfo{year}{2015}\natexlab{b}.
\newblock \showarticletitle{Mining high utility itemsets for transaction
  deletion in a dynamic database}.
\newblock \bibinfo{journal}{\emph{Intelligent Data Analysis}}
  \bibinfo{volume}{19}, \bibinfo{number}{1} (\bibinfo{year}{2015}),
  \bibinfo{pages}{43--55}.
\newblock


\bibitem[\protect\citeauthoryear{Lin, Fournier-Viger, and Gan}{Lin
  et~al\mbox{.}}{2016a}]%
        {lin2016fhn}
\bibfield{author}{\bibinfo{person}{Jerry Chun-Wei Lin},
  \bibinfo{person}{Philippe Fournier-Viger}, {and} \bibinfo{person}{Wensheng
  Gan}.} \bibinfo{year}{2016}\natexlab{a}.
\newblock \showarticletitle{{FHN}: An efficient algorithm for mining
  high-utility itemsets with negative unit profits}.
\newblock \bibinfo{journal}{\emph{Knowledge-Based Systems}}
  \bibinfo{volume}{111} (\bibinfo{year}{2016}), \bibinfo{pages}{283--298}.
\newblock


\bibitem[\protect\citeauthoryear{Lin, Gan, Fournier-Viger, Hong, and Chao}{Lin
  et~al\mbox{.}}{2017}]%
        {lin2017fdhup}
\bibfield{author}{\bibinfo{person}{Jerry Chun-Wei Lin},
  \bibinfo{person}{Wensheng Gan}, \bibinfo{person}{Philippe Fournier-Viger},
  \bibinfo{person}{Tzung-Pei Hong}, {and} \bibinfo{person}{Han-Chieh Chao}.}
  \bibinfo{year}{2017}\natexlab{}.
\newblock \showarticletitle{{FDHUP}: Fast algorithm for mining discriminative
  high utility patterns}.
\newblock \bibinfo{journal}{\emph{Knowledge and Information Systems}}
  \bibinfo{volume}{51}, \bibinfo{number}{3} (\bibinfo{year}{2017}),
  \bibinfo{pages}{873--909}.
\newblock


\bibitem[\protect\citeauthoryear{Lin, Gan, Fournier-Viger, Hong, and Tseng}{Lin
  et~al\mbox{.}}{2016b}]%
        {2lin2016efficient}
\bibfield{author}{\bibinfo{person}{Jerry Chun-Wei Lin},
  \bibinfo{person}{Wensheng Gan}, \bibinfo{person}{Philippe Fournier-Viger},
  \bibinfo{person}{Tzung-Pei Hong}, {and} \bibinfo{person}{Vincent~S Tseng}.}
  \bibinfo{year}{2016}\natexlab{b}.
\newblock \showarticletitle{Efficient algorithms for mining high-utility
  itemsets in uncertain databases}.
\newblock \bibinfo{journal}{\emph{Knowledge-Based Systems}}
  \bibinfo{volume}{96} (\bibinfo{year}{2016}), \bibinfo{pages}{171--187}.
\newblock


\bibitem[\protect\citeauthoryear{Lin, Gan, Fournier-Viger, Hong, and Tseng}{Lin
  et~al\mbox{.}}{2016c}]%
        {lin2016fast}
\bibfield{author}{\bibinfo{person}{Jerry Chun-Wei Lin},
  \bibinfo{person}{Wensheng Gan}, \bibinfo{person}{Philippe Fournier-Viger},
  \bibinfo{person}{Tzung-Pei Hong}, {and} \bibinfo{person}{Vincent~S Tseng}.}
  \bibinfo{year}{2016}\natexlab{c}.
\newblock \showarticletitle{Fast algorithms for mining high-utility itemsets
  with various discount strategies}.
\newblock \bibinfo{journal}{\emph{Advanced Engineering Informatics}}
  \bibinfo{volume}{30}, \bibinfo{number}{2} (\bibinfo{year}{2016}),
  \bibinfo{pages}{109--126}.
\newblock


\bibitem[\protect\citeauthoryear{Lin, Gan, Fournier-Viger, Hong, and Zhan}{Lin
  et~al\mbox{.}}{2016d}]%
        {lin2016efficient}
\bibfield{author}{\bibinfo{person}{Jerry Chun-Wei Lin},
  \bibinfo{person}{Wensheng Gan}, \bibinfo{person}{Philippe Fournier-Viger},
  \bibinfo{person}{Tzung-Pei Hong}, {and} \bibinfo{person}{Justin Zhan}.}
  \bibinfo{year}{2016}\natexlab{d}.
\newblock \showarticletitle{Efficient mining of high-utility itemsets using
  multiple minimum utility thresholds}.
\newblock \bibinfo{journal}{\emph{Knowledge-Based Systems}}
  \bibinfo{volume}{113} (\bibinfo{year}{2016}), \bibinfo{pages}{100--115}.
\newblock


\bibitem[\protect\citeauthoryear{Lin, Gan, Hong, and Tseng}{Lin
  et~al\mbox{.}}{2015a}]%
        {lin2015efficient}
\bibfield{author}{\bibinfo{person}{Jerry Chun-Wei Lin},
  \bibinfo{person}{Wensheng Gan}, \bibinfo{person}{Tzung-Pei Hong}, {and}
  \bibinfo{person}{Vincent~S Tseng}.} \bibinfo{year}{2015}\natexlab{a}.
\newblock \showarticletitle{Efficient algorithms for mining up-to-date
  high-utility patterns}.
\newblock \bibinfo{journal}{\emph{Advanced Engineering Informatics}}
  \bibinfo{volume}{29}, \bibinfo{number}{3} (\bibinfo{year}{2015}),
  \bibinfo{pages}{648--661}.
\newblock


\bibitem[\protect\citeauthoryear{Liu, Hsu, and Ma}{Liu et~al\mbox{.}}{1999}]%
        {liu1999mining}
\bibfield{author}{\bibinfo{person}{Bing Liu}, \bibinfo{person}{Wynne Hsu},
  {and} \bibinfo{person}{Yiming Ma}.} \bibinfo{year}{1999}\natexlab{}.
\newblock \showarticletitle{Mining association rules with multiple minimum
  supports}. In \bibinfo{booktitle}{\emph{Proceedings of the Fifth ACM SIGKDD
  International Conference on Knowledge Discovery and Data Mining}}. ACM,
  \bibinfo{pages}{337--341}.
\newblock


\bibitem[\protect\citeauthoryear{Liu, Wang, and Fung}{Liu
  et~al\mbox{.}}{2012}]%
        {liu2012direct}
\bibfield{author}{\bibinfo{person}{Junqiang Liu}, \bibinfo{person}{Ke Wang},
  {and} \bibinfo{person}{Benjamin~CM Fung}.} \bibinfo{year}{2012}\natexlab{}.
\newblock \showarticletitle{Direct discovery of high utility itemsets without
  candidate generation}. In \bibinfo{booktitle}{\emph{Proceedings of the IEEE
  12th International Conference on Data Mining}}. IEEE,
  \bibinfo{pages}{984--989}.
\newblock


\bibitem[\protect\citeauthoryear{Liu and Qu}{Liu and Qu}{2012}]%
        {liu2012mining}
\bibfield{author}{\bibinfo{person}{Mengchi Liu} {and} \bibinfo{person}{Junfeng
  Qu}.} \bibinfo{year}{2012}\natexlab{}.
\newblock \showarticletitle{Mining high utility itemsets without candidate
  generation}. In \bibinfo{booktitle}{\emph{Proceedings of the 21st ACM
  International Conference on Information and Knowledge Management}}. ACM,
  \bibinfo{pages}{55--64}.
\newblock


\bibitem[\protect\citeauthoryear{Liu, Liao, and Choudhary}{Liu
  et~al\mbox{.}}{2005}]%
        {liu2005two}
\bibfield{author}{\bibinfo{person}{Ying Liu}, \bibinfo{person}{Wei-keng Liao},
  {and} \bibinfo{person}{Alok Choudhary}.} \bibinfo{year}{2005}\natexlab{}.
\newblock \showarticletitle{A two-phase algorithm for fast discovery of high
  utility itemsets}. In \bibinfo{booktitle}{\emph{Pacific-Asia Conference on
  Knowledge Discovery and Data Mining}}. Springer, \bibinfo{pages}{689--695}.
\newblock


\bibitem[\protect\citeauthoryear{Liu, Cheng, and Tseng}{Liu
  et~al\mbox{.}}{2011}]%
        {liu2011discovering}
\bibfield{author}{\bibinfo{person}{Yu-Cheng Liu}, \bibinfo{person}{Chun-Pei
  Cheng}, {and} \bibinfo{person}{Vincent~S Tseng}.}
  \bibinfo{year}{2011}\natexlab{}.
\newblock \showarticletitle{Discovering relational-based association rules with
  multiple minimum supports on microarray datasets}.
\newblock \bibinfo{journal}{\emph{Bioinformatics}} \bibinfo{volume}{27},
  \bibinfo{number}{22} (\bibinfo{year}{2011}), \bibinfo{pages}{3142--3148}.
\newblock


\bibitem[\protect\citeauthoryear{Marshall}{Marshall}{1926}]%
        {marshall1926principles}
\bibfield{author}{\bibinfo{person}{A. Marshall}.}
  \bibinfo{year}{1926}\natexlab{}.
\newblock \showarticletitle{Principles of Economics}. In
  \bibinfo{booktitle}{\emph{Eighth ed. Macmillan and Co.}} London.
\newblock


\bibitem[\protect\citeauthoryear{Pasquier, Bastide, Taouil, and
  Lakhal}{Pasquier et~al\mbox{.}}{1998}]%
        {pasquier1998pruning}
\bibfield{author}{\bibinfo{person}{Nicolas Pasquier}, \bibinfo{person}{Yves
  Bastide}, \bibinfo{person}{Rafik Taouil}, {and} \bibinfo{person}{Lotfi
  Lakhal}.} \bibinfo{year}{1998}\natexlab{}.
\newblock \showarticletitle{Pruning closed itemset lattices for association
  rules}. In \bibinfo{booktitle}{\emph{International Conference on Advanced
  Databases}}. \bibinfo{pages}{177--196}.
\newblock


\bibitem[\protect\citeauthoryear{Rymon}{Rymon}{1992}]%
        {rymon1992search}
\bibfield{author}{\bibinfo{person}{Ron Rymon}.}
  \bibinfo{year}{1992}\natexlab{}.
\newblock \showarticletitle{Search through systematic set enumeration}.
\newblock  (\bibinfo{year}{1992}).
\newblock


\bibitem[\protect\citeauthoryear{Song, Zhang, and Li}{Song
  et~al\mbox{.}}{2016}]%
        {song2016high}
\bibfield{author}{\bibinfo{person}{Wei Song}, \bibinfo{person}{Zihan Zhang},
  {and} \bibinfo{person}{Jinhong Li}.} \bibinfo{year}{2016}\natexlab{}.
\newblock \showarticletitle{A high utility itemset mining algorithm based on
  subsume index}.
\newblock \bibinfo{journal}{\emph{Knowledge and Information Systems}}
  \bibinfo{volume}{49}, \bibinfo{number}{1} (\bibinfo{year}{2016}),
  \bibinfo{pages}{315--340}.
\newblock


\bibitem[\protect\citeauthoryear{Srikant and Agrawal}{Srikant and
  Agrawal}{1996}]%
        {srikant1996mining}
\bibfield{author}{\bibinfo{person}{Ramakrishnan Srikant} {and}
  \bibinfo{person}{Rakesh Agrawal}.} \bibinfo{year}{1996}\natexlab{}.
\newblock \showarticletitle{Mining sequential patterns: Generalizations and
  performance improvements}. In \bibinfo{booktitle}{\emph{International
  Conference on Extending Database Technology}}. Springer,
  \bibinfo{pages}{1--17}.
\newblock


\bibitem[\protect\citeauthoryear{Tseng, Shie, Wu, and Philip}{Tseng
  et~al\mbox{.}}{2013}]%
        {tseng2013efficient}
\bibfield{author}{\bibinfo{person}{Vincent~S Tseng}, \bibinfo{person}{Bai-En
  Shie}, \bibinfo{person}{Cheng-Wei Wu}, {and} \bibinfo{person}{S~Yu Philip}.}
  \bibinfo{year}{2013}\natexlab{}.
\newblock \showarticletitle{Efficient algorithms for mining high utility
  itemsets from transactional databases}.
\newblock \bibinfo{journal}{\emph{IEEE Transactions on Knowledge and Data
  Engineering}} \bibinfo{volume}{25}, \bibinfo{number}{8}
  (\bibinfo{year}{2013}), \bibinfo{pages}{1772--1786}.
\newblock


\bibitem[\protect\citeauthoryear{Tseng, Wu, Fournier-Viger, and Philip}{Tseng
  et~al\mbox{.}}{2016}]%
        {tseng2016efficient}
\bibfield{author}{\bibinfo{person}{Vincent~S Tseng}, \bibinfo{person}{Cheng-Wei
  Wu}, \bibinfo{person}{Philippe Fournier-Viger}, {and} \bibinfo{person}{S~Yu
  Philip}.} \bibinfo{year}{2016}\natexlab{}.
\newblock \showarticletitle{Efficient algorithms for mining top-$k$ high
  utility itemsets}.
\newblock \bibinfo{journal}{\emph{IEEE Transactions on Knowledge and Data
  Engineering}} \bibinfo{volume}{28}, \bibinfo{number}{1}
  (\bibinfo{year}{2016}), \bibinfo{pages}{54--67}.
\newblock


\bibitem[\protect\citeauthoryear{Tseng, Wu, Shie, and Yu}{Tseng
  et~al\mbox{.}}{2010}]%
        {tseng2010up}
\bibfield{author}{\bibinfo{person}{Vincent~S Tseng}, \bibinfo{person}{Cheng-Wei
  Wu}, \bibinfo{person}{Bai-En Shie}, {and} \bibinfo{person}{Philip~S Yu}.}
  \bibinfo{year}{2010}\natexlab{}.
\newblock \showarticletitle{{UP-G}rowth: an efficient algorithm for high
  utility itemset mining}. In \bibinfo{booktitle}{\emph{Proceedings of the 16th
  ACM SIGKDD International Conference on Knowledge Discovery and Data Mining}}.
  ACM, \bibinfo{pages}{253--262}.
\newblock


\bibitem[\protect\citeauthoryear{Vo, Hong, and Le}{Vo et~al\mbox{.}}{2013}]%
        {vo2013lattice}
\bibfield{author}{\bibinfo{person}{Bay Vo}, \bibinfo{person}{Tzung-Pei Hong},
  {and} \bibinfo{person}{Bac Le}.} \bibinfo{year}{2013}\natexlab{}.
\newblock \showarticletitle{A lattice-based approach for mining most
  generalization association rules}.
\newblock \bibinfo{journal}{\emph{Knowledge-Based Systems}}
  \bibinfo{volume}{45} (\bibinfo{year}{2013}), \bibinfo{pages}{20--30}.
\newblock


\bibitem[\protect\citeauthoryear{Vo, Le, Hong, and Le}{Vo
  et~al\mbox{.}}{2015}]%
        {vo2015fast}
\bibfield{author}{\bibinfo{person}{Bay Vo}, \bibinfo{person}{Tuong Le},
  \bibinfo{person}{Tzung-Pei Hong}, {and} \bibinfo{person}{Bac Le}.}
  \bibinfo{year}{2015}\natexlab{}.
\newblock \showarticletitle{Fast updated frequent-itemset lattice for
  transaction deletion}.
\newblock \bibinfo{journal}{\emph{Data \& Knowledge Engineering}}
  \bibinfo{volume}{96} (\bibinfo{year}{2015}), \bibinfo{pages}{78--89}.
\newblock


\bibitem[\protect\citeauthoryear{Wang and Huang}{Wang and Huang}{2018}]%
        {wang2018incremental}
\bibfield{author}{\bibinfo{person}{Jun-Zhe Wang} {and}
  \bibinfo{person}{Jiun-Long Huang}.} \bibinfo{year}{2018}\natexlab{}.
\newblock \showarticletitle{On Incremental High Utility Sequential Pattern
  Mining}.
\newblock \bibinfo{journal}{\emph{ACM Transactions on Intelligent Systems and
  Technology}} \bibinfo{volume}{9}, \bibinfo{number}{5} (\bibinfo{year}{2018}),
  \bibinfo{pages}{55}.
\newblock


\bibitem[\protect\citeauthoryear{Wang, Huang, and Chen}{Wang
  et~al\mbox{.}}{2016}]%
        {wang2016efficiently}
\bibfield{author}{\bibinfo{person}{Jun-Zhe Wang}, \bibinfo{person}{Jiun-Long
  Huang}, {and} \bibinfo{person}{Yi-Cheng Chen}.}
  \bibinfo{year}{2016}\natexlab{}.
\newblock \showarticletitle{On efficiently mining high utility sequential
  patterns}.
\newblock \bibinfo{journal}{\emph{Knowledge and Information Systems}}
  \bibinfo{volume}{49}, \bibinfo{number}{2} (\bibinfo{year}{2016}),
  \bibinfo{pages}{597--627}.
\newblock


\bibitem[\protect\citeauthoryear{Weiss}{Weiss}{2004}]%
        {weiss2004mining}
\bibfield{author}{\bibinfo{person}{Gary~M Weiss}.}
  \bibinfo{year}{2004}\natexlab{}.
\newblock \showarticletitle{Mining with rarity: a unifying framework}.
\newblock \bibinfo{journal}{\emph{ACM SIGKDD Explorations Newsletter}}
  \bibinfo{volume}{6}, \bibinfo{number}{1} (\bibinfo{year}{2004}),
  \bibinfo{pages}{7--19}.
\newblock


\bibitem[\protect\citeauthoryear{Yao and Hamilton}{Yao and Hamilton}{2006}]%
        {yao2006mining}
\bibfield{author}{\bibinfo{person}{Hong Yao} {and} \bibinfo{person}{Howard~J
  Hamilton}.} \bibinfo{year}{2006}\natexlab{}.
\newblock \showarticletitle{Mining itemset utilities from transaction
  databases}.
\newblock \bibinfo{journal}{\emph{Data \& Knowledge Engineering}}
  \bibinfo{volume}{59}, \bibinfo{number}{3} (\bibinfo{year}{2006}),
  \bibinfo{pages}{603--626}.
\newblock


\bibitem[\protect\citeauthoryear{Yao, Hamilton, and Butz}{Yao
  et~al\mbox{.}}{2004}]%
        {yao2004foundational}
\bibfield{author}{\bibinfo{person}{Hong Yao}, \bibinfo{person}{Howard~J
  Hamilton}, {and} \bibinfo{person}{Cory~J Butz}.}
  \bibinfo{year}{2004}\natexlab{}.
\newblock \showarticletitle{A foundational approach to mining itemset utilities
  from databases}. In \bibinfo{booktitle}{\emph{Proceedings of the SIAM
  International Conference on Data Mining}}. SIAM, \bibinfo{pages}{482--486}.
\newblock


\bibitem[\protect\citeauthoryear{Yin, Zheng, and Cao}{Yin
  et~al\mbox{.}}{2012}]%
        {yin2012uspan}
\bibfield{author}{\bibinfo{person}{Junfu Yin}, \bibinfo{person}{Zhigang Zheng},
  {and} \bibinfo{person}{Longbing Cao}.} \bibinfo{year}{2012}\natexlab{}.
\newblock \showarticletitle{{US}pan: an efficient algorithm for mining high
  utility sequential patterns}. In \bibinfo{booktitle}{\emph{Proceedings of the
  18th ACM SIGKDD International Conference on Knowledge Discovery and Data
  Mining}}. ACM, \bibinfo{pages}{660--668}.
\newblock


\bibitem[\protect\citeauthoryear{Zida, Fournier-Viger, Lin, Wu, and Tseng}{Zida
  et~al\mbox{.}}{2015}]%
        {zida2015efim}
\bibfield{author}{\bibinfo{person}{Souleymane Zida}, \bibinfo{person}{Philippe
  Fournier-Viger}, \bibinfo{person}{Jerry Chun-Wei Lin},
  \bibinfo{person}{Cheng-Wei Wu}, {and} \bibinfo{person}{Vincent~S Tseng}.}
  \bibinfo{year}{2015}\natexlab{}.
\newblock \showarticletitle{{EFIM}: a highly efficient algorithm for
  high-utility itemset mining}. In \bibinfo{booktitle}{\emph{Mexican
  International Conference on Artificial Intelligence}}. Springer,
  \bibinfo{pages}{530--546}.
\newblock


\end{thebibliography}

\end{document}